\documentclass[journal,10pt]{IEEEtran}
\usepackage{color}
\usepackage{mathrsfs}
\usepackage{amsfonts}
\usepackage{amsmath}
\usepackage{amssymb}
\usepackage{graphicx}
\usepackage{subfigure}
\IEEEoverridecommandlockouts

\newtheorem{definition}{Definition}
\newtheorem{theorem}{Theorem}
\newtheorem{lemma}{Lemma}
\newtheorem{corollary}{Corollary}
\newtheorem{remark}{Remark}
\newtheorem{assumption}{Assumption}
\newtheorem{problem}{Problem}
\newtheorem{algorithm}{Algorithm}
\newcommand{\tcb}{\textcolor{black}}

\allowdisplaybreaks[4]
\begin{document}

\title{Dynamic Power Control for Delay-Aware Device-to-Device Communications}
%\author
%{\IEEEauthorblockN{Wei Wang,~\IEEEmembership{Member,~IEEE}$^\fnsymbol{one}$$^\fnsymbol{two}$, Fan Zhang,~\IEEEmembership{Student Member,~IEEE}$^\fnsymbol{one}$$^\fnsymbol{two}$, Vincent K. N. Lau,~\IEEEmembership{Fellow,~IEEE}$^\fnsymbol{one}$$^\fnsymbol{two}$}
%\IEEEauthorblockA{Department of Electronic and Computer Engineering, Hong Kong University of Science and Technology, Hong Kong.}\thanks{The authors are with Department of Electronic and Computer Engineering, Hong Kong University of Science and Technology, Hong Kong. (Email:\{eewangw, fzhangee, eeknlau\}@ust.hk)}
%}

\author{\IEEEauthorblockN{Wei Wang,~\IEEEmembership{Member,~IEEE}$^\dag$$^\ddag$, Fan Zhang,~\IEEEmembership{Student Member,~IEEE}$^\ddag$, Vincent K. N. Lau,~\IEEEmembership{Fellow,~IEEE}$^\ddag$
\thanks{Manuscript received May 3, 2013; revised November 14, 2013; accepted
August 6, 2014. This work is supported in part by National Natural Science Foundation of China (Nos. 61261130585, 61001098), National Hi-Tech R\&D Program (No. 2014AA01A702), Research Grants Council (RGC) of Hong Kong (No. N$_-$HKUST605/13), and Hong Kong Scholars Program (No. 2012T50566). Vincent K. N. Lau is also the Changjiang Chair Professor at Zhejiang University, P.R.~China.}\\
}
\vspace{0.2cm}
\parbox{0.55\textwidth}{\centering $^\dag$Dept. of Information Science and Electronic Engineering \\
Zhejiang Key Lab. of Information Network Technology \\
Zhejiang University, Hangzhou 310027, P.R. China}
\hfill
\parbox{0.44\textwidth}{\centering $^\ddag$Dept. of Electrical and Computer Engineering\\
Hong Kong University of Science and Technology\\
Clear Water Bay, Hong Kong}
Email: wangw@zju.edu.cn, fzhangee@ust.hk, eeknlau@ust.hk
}

%\addtolength{\abovedisplayskip}{-0.05ex}
%\addtolength{\belowdisplayskip}{-0.05ex}
%\addtolength{\abovedisplayshortskip}{-0.05ex}
%\addtolength{\belowdisplayshortskip}{-0.05ex}

\maketitle

%\vspace{-1.2cm}
\begin{abstract}
In this paper, we consider the dynamic power control for  delay-aware D2D communications. The stochastic optimization problem is formulated as an infinite horizon average cost Markov decision process. To deal with the curse of dimensionality, we utilize the interference filtering property of the CSMA-like MAC protocol and derive a closed-form approximate priority function  and the associated error bound using perturbation analysis. Based on the closed-form approximate priority function, we propose a low-complexity power control algorithm solving  the per-stage optimization problem. The proposed solution is further shown to be asymptotically optimal for a sufficiently large carrier sensing distance.  Finally, the proposed power control scheme is compared with various baselines through simulations, and it is shown that significant performance gain can be achieved.
\end{abstract}

\IEEEpeerreviewmaketitle

\section{Introduction}
Future wireless cellular networks (e.g. IMT-advanced) are expected to provide higher data rates and system capacity. One potential technology to meet the demands is the infrastructure-assisted device-to-device (D2D) communications \cite{3gpp}. Taking advantage of the physical proximity of communication devices, the D2D technique enables direct communications between devices, which results in high data rates, low delays and low power consumption. Unlike conventional ad hoc networks, the cellular base station (BS) plays an important role for D2D communications in helping the D2D nodes on both peer discovery and resource allocation \cite{D2D1}. There are several existing works on D2D communications in cellular networks. In \cite{D2D2} and \cite{D2D3}, the D2D nodes share the spectrum with cellular users using an underlay approach, in which the throughput of D2D communications is maximized while the QoS of the cellular users is guaranteed. In \cite{D2D4} and \cite{D2D5}, the maximum sum-rate of the network is achieved by dynamically selecting one of the transmission modes, including D2D mode with shared channels, D2D mode with dedicated channels and cellular transmission mode. In \cite{D2D6}, the multi-antenna cellular BS acts as a cooperative relay, helping the D2D nodes forward packets so as to improve the throughput of the network. {Power control is important for interference coordination among the nodes in wireless networks. The transmit power is adjusted to meet the users' required signal to interference plus noise ratios (SINR)~\cite{power1}, satisfy the received signal power level~\cite{power3} or achieve a higher data rate~\cite{power4}. In \cite{power5}, the transmit power is minimized for D2D communications subject to a sum-rate constraint.}
However, these existing works have all focused on the physical layer performance without consideration of the bursty data arrivals at the transmitters as well as the delay requirement of the information flows. Since real-life applications (such as video streaming, web browsing or VoIP) are delay-sensitive, it is  important to optimize the delay performance for D2D communications.

To take the queueing delay into consideration, the radio resource control policy should be a function of both the channel state information (CSI) and the queue state information (QSI). This is because the CSI reveals the instantaneous transmission opportunities at the physical layer and the QSI reveals the urgency of the data flows. However, the associated optimization problem is very challenging. A systematic approach to the delay-aware optimization problem is through the Markov Decision Process (MDP). In general, the optimal control policy can be obtained by solving the well-known \emph{Bellman equation}. Conventional solutions to the Bellman equation, such as brute-force value iteration or policy iteration \cite{DPcontrol}, have huge complexity (i.e., the curse of dimensionality), because solving the Bellman equation involves solving an exponentially large system of non-linear equations. There are some existing works that use the \emph{stochastic approximation} approach with distributed online learning algorithm \cite{delay2}, which has linear complexity. However, the stochastic learning approach can only give a numerical solution to the Bellman equation and may suffer from slow convergence and lack of insight. We treat this issue and provide some preliminary results on cross-layer design with closed-form solution in \cite{mag}.

In this paper, we investigate the dynamic power control for D2D communications systems. We focus on minimizing the average transmit power and the average delay of  the D2D data flows. There are several technical challenges associated with the dynamic power control optimization problem.
\begin{itemize}
\item \textbf{Challenges due to the Average Delay Consideration:} Unlike other papers which optimize the physical layer throughput of the D2D systems, the optimization involving delay constraints is fundamentally challenging. This is because the associated problem belongs to the class of {\em stochastic optimization} \cite{delay1}, which embraces both {\em information theory} (to model the physical layer dynamics) and {\em queueing theory} (to model the queue dynamics). A key obstacle to solving the associated Bellman equation is to obtain the priority function, and there is no easy and systematic solution in general \cite{DPcontrol}.
\item \textbf{Challenges due to the Coupled Queue Dynamics:} {The interference among the D2D nodes \cite{interference1,interference3} fundamentally induces coupled queue dynamics among the D2D flows. For instance, the service rate of the queue for each D2D flow depends on the transmit power of all the other active D2D flows due to the mutual interference.} The associated stochastic optimization problem is a $K$-dimensional MDP, where $K$ is the number of D2D flows. This $K$-dimensional MDP leads to the curse of dimensionality with complexity exponential to $K$ for solving the associated Bellman equation. It is  highly nontrivial to obtain a low complexity solution for the dynamic resource control of the D2D systems.
\item \textbf{Challenges due to the Non-Convexity Nature:} Despite the complexity issue involved in obtaining the priority function for the stochastic optimization problem, the per-stage control optimization in the Bellman equation is also non-convex due to the mutual interference term in the mutual information. This poses a great challenge in solving the delay-constrained optimization in the D2D systems.
\end{itemize}

\begin{figure}
  \centering
  \includegraphics[width=3.5in]{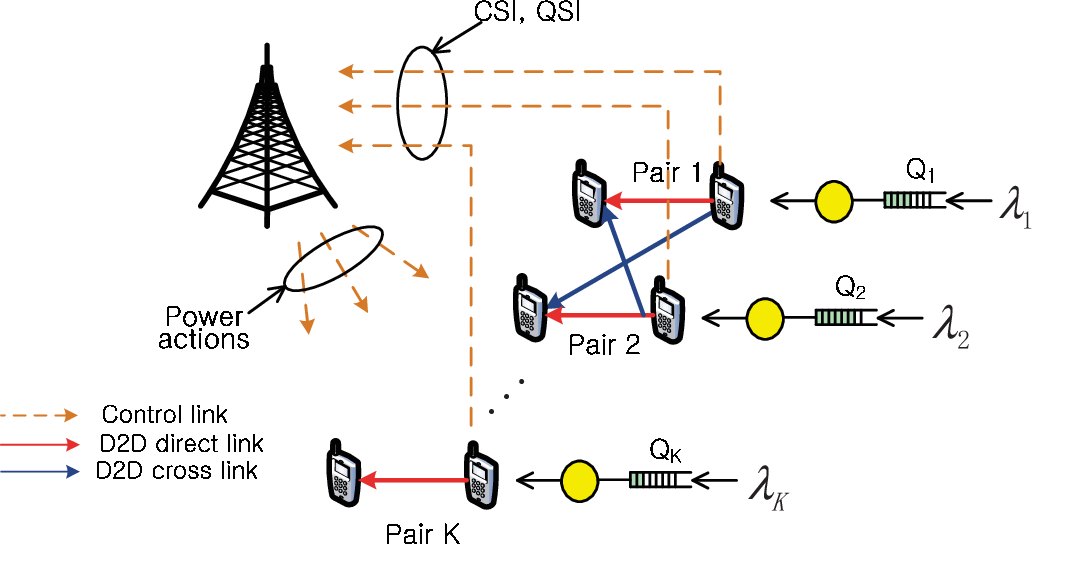}
  \caption{Topology of an infrastructure-assisted D2D communications system.}
  \label{topo}	\vspace{-0.3cm}
\end{figure}

In this paper, we first establish the PHY, MAC and bursty data source models as well as the queue dynamics in Section II. We formally formulate the associated stochastic optimization problem of the dynamic power control for delay-aware D2D communications as {\em an infinite horizon average cost MDP}. {To overcome the aforementioned technical challenges, we exploit specific problem structures in D2D communications. Specifically, 1) the CSMA-like MAC protocol is adopted to coordinate the transmissions of the D2D nodes in a distributive way and this induces a \emph{weak interference} topology among the simultaneously transmitting D2D nodes, and 2) the assistance of the BS substantially simplifies the signaling mechanism of control information exchange.} We derive a {\em simplified optimality condition} for solving the MDP in Section III. Compared with the conventional  {\em Bellman equation} \cite{DPcontrol}, the derived optimality condition involves solving a $K$-dimensional partial differential equation (PDE) only. Utilizing the {\em interference filtering property} of the MAC protocol, we obtain a closed-form approximate priority function and the associated error bound using {\em perturbation analysis}. Based on that, we obtain a delay-aware low complexity dynamic power control algorithm for the D2D communications in Section IV. The solution is shown to be asymptotically optimal for a sufficiently large carrier sensing distance in the MAC protocol. Furthermore, in Section V, we show that the proposed solution achieves significant performance gain over various baseline schemes.

\section{System Model}
In this section, we introduce the system model for the infrastructure-assisted D2D communications, including the D2D system topology, the physical layer model, the MAC layer model and the bursty data source model. {We first list the important notations in this paper in Table~1.}
\begin{table}[!h]
\caption{List of important notations}
\vspace{-0.3cm}
\begin{center}
\begin{tabular}{l@{}ll@{}l}
\hline
\multicolumn{2}{c}{Symbol}&\multicolumn{2}{c}{Meaning} \\
\hline
$K$                  &&number of D2D pairs\\
$\mathbf{P}=\{P_{k}\}$                &&transmit power\\
$\mathbf{H}=\{H_{kj}\}$                &&global CSI\\
$\mathbf{L}=\{L_{kj}\}$                &&large-scale path gain\\
$\boldsymbol{\sigma}=\{\sigma_{k}\}$                &&MAC output\\
$\boldsymbol{\nu}=\{\nu_{k}\}$                &&probability of accessing the channel\\
$\mathbf{A}=\{A_{k}\}$                &&bit/packet arrival\\
$\boldsymbol{\lambda}=\{\lambda_{k}\}$                &&average arrival rate\\
$\mathbf{Q}=\{Q_{k}\}$                &&global QSI\\
$\boldsymbol{\chi}=\{\boldsymbol{\sigma},\mathbf{H},\mathbf{Q}\}$  &&global system state\\
$\boldsymbol{\Omega}(\boldsymbol{\chi})=\{\Omega_{k}(\boldsymbol{\chi})\}$  &&power control policy\\
$\tau$               &&duration of a time slot\\
$C_k(\mathbf{H},\mathbf{P})$                 &&achievable data rate of the $k$-th D2D pair&\\
$\delta$               &&carrier sensing distance\\
$L^\delta$           &&worst-case cross-channel path gain\\
$V^*(\mathbf{Q})$    &&priority function\\
\hline
\end{tabular}
\end{center}
\vspace{-0.3cm}
\end{table}

\subsection{D2D System Topology}
We consider an infrastructure-assisted D2D communications system, as shown in Fig.~\ref{topo}. Specifically, the D2D system consists of two tiers, namely the {\em cellular tier} and the {\em D2D tier}. In the D2D tier, there are $K$ transmitter-receiver (Tx-Rx) pairs located randomly in the area of a cell. Transmitter $k$ transmits data to receiver $k$, and the Tx-Rx pair is associated by the D2D peer discovery procedure \cite{D2D1}. All D2D pairs share a common channel, which is orthogonal to the channels used in the cellular tier\footnote{The channel for D2D communications could be a dedicated part of the licensed spectrum allocated by the BS, or another spectrum band, e.g., Wifi D2D transmission on the ISM band.}. Hence, there is no cross-tier interference between the cellular and D2D tiers. In the cellular tier, the BS plays the role of the centralized controller for the D2D communications. Each D2D pair communicates directly on a single-hop link in a distributed ad-hoc manner with the assistance of the cellular BS. The time is slotted, and the duration of each time slot is $\tau$. The cellular BS collects  necessary information and broadcasts the resource allocation actions (calculated based on the collected information) periodically to the D2D nodes at the beginning of each time slot.

\subsection{Physical Layer Model}
Let $s_k$ denote the information symbol for the $k$-th D2D pair. The received signal at receiver $k$ is

\begin{equation}
y_k= \underbrace{h_{kk} \sqrt{P_k} s_k}_{\text{desired signal}}+\underbrace{\sum_{j\neq k} h_{kj}\sqrt{P_j} s_j}_{\text{interference}}+\underbrace{z_k}_{\text{noise}}
\end{equation}
where $h_{kj}$ is the complex channel fading coefficient between transmitter $j$ and receiver $k$, and $z_k\sim \mathcal{CN}(0,N_0)$ is the i.i.d. complex Gaussian channel noise with  power $N_0$. $P_k$ is the transmitter power for $s_k$.  Let $\mathbf{H}(t)=\{H_{kj}(t):\forall j,k \}$ be the global CSI, where $H_{kj}(t)=|h_{kj}(t)|^2$ is the instantaneous channel path gain from transmitter $j$ to receiver $k$ at the $t$-th time slot. {We consider the CSI according to the block fading channel model \cite{block1,block2} and have the following assumption on $\mathbf{H}$:}
\begin{assumption}[Short-Term CSI Model]\label{a_CSI}
The CSI $\mathbf{H}(t)$ remains constant within a time slot and is i.i.d. over time slots. $H_{kj}(t)$ follows a negative exponential distribution\footnote{{Rayleigh fading is adopted as an example here for algebraic simplicity. The proposed optimization framework is general to cover various channel fading models as well. With other fading models, the difference is in integrating with different fading distributions when calculating the expectation over $\mathbf{H}$ to estimate the expected future cost.}} with mean $L_{kj}$.  Furthermore, $H_{kj}(t)$ is independent w.r.t. the D2D pair indices $k,j$.~\hfill \IEEEQED
\end{assumption}

Note that $L_{kj}$ is the  large-scale path gain from transmitter $j$ to receiver $k$. Let $\mathbf{L}=\{L_{kj}:\forall j,k\}$, and we have the following assumption on $\mathbf{L}$.

\begin{assumption}  [Long-Term Path Gain Model] \label{long}
The long-term large-scale path gain $\mathbf{L}$ is constant for the duration of the communication session. Specifically, for any transmitter $j$ and receiver $k$,  the relationship between the  path gain $L_{kj}$ and the distance $d_{kj}$ is\footnote{Here we adopt the Friis free space path loss model \cite{tse}. Note that the results of this paper can be extended easily for other path loss models.}  $L_{kj}= \frac{G_k^r G_j^t \lambda^2}{\left( 4 \pi d_{kj}\right)^2}$ ($\forall k,j$), where  $G_k^r$ and $G_j^t$ are the receive and transmit antenna gains respectively, and $\lambda$ is the carrier wavelength.~\hfill \IEEEQED
\end{assumption}

Let $\mathbf{P}(t)=\{P_k(t): \forall k\}$ be the collection of the transmit power of all the D2D transmitters at the $t$-th time slot. {For given CSI $\mathbf{H}(t)$ and power actions $\mathbf{P}(t)$, the achievable data rate of the $k$-th Tx-Rx pair depends on the SINR by treating interference as noise, which is calculated as}
\begin{equation}\label{capacity}
C_k(\mathbf{H}(t), \mathbf{P}(t))=\log_2\left(1+{\frac{1}{\Gamma}}\frac{H_{kk}(t)P_k(t)}{N_0+\sum_{j\neq k}H_{kj}(t)P_j(t)}\right)
\end{equation}
{where $\Gamma$ is the SINR gap \cite{SINRgap} to measure the practical reduction of the SINR with respect to the capacity. $\Gamma$ depends on the error probability requirement as well as the modulation scheme.}

\subsection{MAC Layer Model}
The D2D nodes utilize a CSMA-like protocol to arbitrate the random channel access in a distributed manner. The basic principle of the CSMA is {\em listen-before-talk} \cite{csma}, which is used  to avoid collision between simultaneous transmissions of neighboring nodes. As a result, the MAC protocol  determines the subset of the D2D nodes in which the transmitters can  transmit data simultaneously without causing excessive interference. For simplicity, we consider the following idealized MAC protocol model, which has  been widely adopted in justifying the {\em hardcore point process} \cite{geo}.
\begin{assumption}[Hardcore Point Process Model]\label{a_MAC}
The D2D nodes adopt a CSMA-like MAC protocol with the {\em carrier sensing distance}\footnote{Carrier sensing distance refers to the carrier sensing range of the associated CSMA protocol. Two nodes within the carrier sensing distance will not transmit simultaneously.} $\delta$. {The output of the MAC protocol is captured by the {\em MAC output process} $\boldsymbol{\sigma}(t) = (\sigma_1(t), \cdots,\sigma_K(t)) \in \{0,1\}^K$, where $\sigma_k(t)=1$ means that the $k$-th D2D node accesses the channel at the $t$-th time slot. The MAC output process $\boldsymbol{\sigma}(t)$ has the following properties:}
\begin{itemize}
\item ${\sigma}_k(t)$ is i.i.d. over time slots according to the Bernoulli distribution with mean $\mathbb{E}[\sigma_k(t)]=\nu_k\left(\delta \right)$.
\item all transmit nodes have equal opportunity to access the channel, i.e., $\nu_k\left(\delta \right)=\frac{1}{\left|\mathcal{N}_k\left(\delta \right)\right|+1}$, where $\mathcal{N}_k\left(\delta \right)$ is the set of transmit nodes within the carrier sensing distance $\delta$ from transmitter $k$ and  $\left|\mathcal{N}_k\left(\delta \right)\right|$ is the associated cardinality.
\item during each time slot $t$, a feasible $\boldsymbol{\sigma}(t)$ satisfies the following carrier sensing constraint: if $\sigma_k(t)=1$,  then $\sigma_j(t)=0$ for all $j\in \mathcal{N}_k$.
\hfill \IEEEQED
\end{itemize}
\end{assumption}

The first condition corresponds to the memoryless property of the MAC protocol with respect to the channel access. The second condition corresponds to the fairness among the D2D nodes in the neighbour set, and the third condition corresponds to the carrier sensing requirement in the MAC protocol. Note that $\nu_k(\delta)$ corresponds to the \emph{spatial reuse factor} for transmitter $k$ in the D2D network for a given carrier sensing distance $\delta$. Furthermore,  $\nu_k\left(\delta \right)$ and $\mathcal{N}_k \left(\delta \right)$ depend on the topology of the D2D nodes.

\subsection{Bursty Data Source and Queue Dynamics}

There is a bursty data source at each D2D transmitter. Let $\mathbf{A}(t)=(A_1(t)\tau, \cdots, A_K(t)\tau)$ be the random arrivals (number of bits) from the application layers to the $K$ D2D transmitters at the end of the $t$-th time slot\footnote{We assume that the transmitters are causal so that the packets arrived at the time slot are not observed when the control actions of this time slot are performed.}. We have the following assumption on $\mathbf{A}(t)$.
\begin{assumption}[Bursty Source Model]
Assume that $A_k\left(t\right)$ is i.i.d. over decision slots according to a general distribution $\Pr[A_k]$. The moment generating function of $A_k$ exists with  $\mathbb{E}[A_k]=\lambda_k$.  $ A_k\left(t\right)$ is independent w.r.t. $k$. Furthermore, the arrival rates  $(\lambda_1,\dots,\lambda_K)$ lie within the stability region \cite{comparison2} of the  system.~\hfill \IEEEQED
\end{assumption}

Each D2D transmitter has a data queue for the bursty traffic flows towards the associated receiver. Let $Q_k(t)\in [0, \infty)$ be the queue length (number of bits) at  transmitter $k$ at the beginning of the $t$-th  slot. Let $\mathbf{Q}(t)=(Q_1(t),\cdots,Q_K(t))\in \boldsymbol{\mathcal{Q}}\triangleq [0, \infty)^K$ be the global QSI. The queue dynamics of  transmitter $k$  is
\begin{equation}\label{queueing}
    Q_k(t+1)=\max \left\{Q_k(t)-\sigma_k(t)C_k(\mathbf{H}(t), \mathbf{P}(t))\tau,0 \right\}+A_k(t)\tau
\end{equation}

\begin{remark}[Weak Coupling Property of Queue Dynamics]
The $K$ queue dynamics in the  D2D system are coupled together due to the interference term in (\ref{capacity}). Specifically, the departure of the queue at each transmitter depends on the power actions of all the $K$ D2D transmitters. Furthermore, the CSMA-like mechanism in the MAC protocol model in Assumption~\ref{a_MAC} contributes to filtering the strong interference between the active D2D transmitters. Let $L^\delta=\max \{L_{kj}: \forall k \neq j, d_{kj}>\delta\}$ be the worst-case cross-channel path gain for a given sensing threshold $\delta$.  Due to the interference filtering property of the MAC protocol, there is only weak queue coupling in the D2D network, and $L^\delta$ measures the \emph{coupling intensity}. We will leverage this weak coupling property to derive low complexity closed-form approximate solutions in Section IV. \hfill \IEEEQED
\end{remark}

\section{Delay-Aware Cross-Layer Control Framework}
In this section, we formally formulate the delay-aware cross-layer radio resource control framework for D2D communications. We first define the control policy and the optimization objective. We then formulate the design as a Markov Decision Process (MDP) and derive the optimality conditions for solving the problem.

\subsection{Power Control Policy}
For delay-sensitive applications, it is important to dynamically adapt the transmit power of the D2D nodes based on the instantaneous realizations of the CSI (captures the instantaneous transmission opportunities) and the QSI (captures the urgency of the $K$ data flows). Let $\boldsymbol{\chi}=(\boldsymbol{\sigma},\mathbf{H},\mathbf{Q})$ denote the global system state. We  define the stationary power control policy below.
\begin{definition}[Stationary Power Control Policy]
A stationary control policy for the $k$-th D2D transmitter $\Omega_k$ is a mapping from the system state $\boldsymbol{\chi}$ to the  power control action of transmitter $k$. Specifically, $\Omega_k(\boldsymbol{\chi})=P_k \geq 0$. Let $\boldsymbol{\Omega}=\{\Omega_k:\forall k\}$ denote the aggregation of the control policies for all the $K$ D2D transmitters.\hfill \IEEEQED
\end{definition}

{Since the D2D nodes access the channel randomly, the MAC output $\boldsymbol{\sigma}$ is i.i.d. over time slots. The CSI $\mathbf{H}$ is i.i.d. over time slots based on the block fading channel model in Assumption~\ref{a_CSI}. Furthermore, from the queue evolution equation in (\ref{queueing}), $\mathbf{Q}(t+1)$ depends only on $\mathbf{Q}(t)$ and the data rate. Given a control policy $\boldsymbol{\Omega}$, the data rate at the $t$-th time slot depends on $\sigma_k(t)$, $\mathbf{H}(t)$ and $\boldsymbol{\Omega}(\boldsymbol{\chi(t)})$. Hence, the global system state ${\boldsymbol{\chi}(t)}$ is a controlled Markov chain \cite{DPcontrol} with the transition probability}
\begin{align}
&\Pr[\boldsymbol{\chi}(t+1)|\boldsymbol{\chi}(t),\boldsymbol{\Omega}(\boldsymbol{\chi}(t))]\\
=&\Pr[\boldsymbol{\sigma}(t+1)]\Pr[\mathbf{H}(t+1)]
\Pr[\mathbf{Q}(t+1)|\boldsymbol{\chi}(t),\boldsymbol{\Omega}(\boldsymbol{\chi}(t))]\nonumber
\end{align}
where the queue transition probability is given by
\begin{align}	
&\Pr[\mathbf{Q}(t+1)|\boldsymbol{\chi}(t),\boldsymbol{\Omega}(\boldsymbol{\chi}(t))]\nonumber\\
=& \left\{
	\begin{aligned}	
		&\prod_{k} \Pr\big[A_k\left( t\right)\big],   \ \ \tcb{\text{if } Q_k\left( t+1\right) \text{is given by (\ref{queueing})},\forall k}\\
		&\ 0, \hspace{2.6cm} \text{otherwise}
	   \end{aligned}
   \right.
  \end{align}
%where the equality is due to the i.i.d. assumptions of $\mathbf{H}(t)$ and $\boldsymbol{\sigma}(t)$ in Assumptions \ref{a_CSI} and \ref{a_MAC}.

For technical reasons, we consider the {\em admissible control policy} defined below.
\begin{definition}[Admissible Control Policy]   \label{admissibledis}
A policy $\boldsymbol{\Omega}$ is admissible if the following requirements are satisfied:
\begin{itemize}
\item $\boldsymbol{\Omega}$ is a unichain policy, i.e., the controlled Markov chain $\left\{\boldsymbol{\chi}\left(t \right)\right\}$ under $\boldsymbol{\Omega}$ has a single recurrent class (and possibly some transient states) \cite{DPcontrol}.

\item The queueing system under $\boldsymbol{\Omega}$ is third-order stable in the sense that $\lim_{t\rightarrow \infty}\mathbb{E}^{\boldsymbol{\Omega}}[\sum_{k=1}^{K}Q_k^3(t)]<\infty$, where $\mathbb{E}^{\boldsymbol{\Omega}}$ means taking expectation w.r.t. the probability measure induced by the control policy $\boldsymbol{\Omega}$.\hfill \IEEEQED
\end{itemize}
\end{definition}

\subsection{Problem Formulation}
As a result, under an admissible control policy $\boldsymbol{\Omega}$, the average delay cost for the $k$-th D2D pair is given by
\begin{align}	\label{delay_cost}
	\overline{D}_k(\boldsymbol{\Omega})  = \limsup_{T \rightarrow \infty} \frac{1}{T} \sum_{t=0}^{T-1} \mathbb{E}^{\boldsymbol{\Omega}} \left[\frac{Q_k\left(t\right)}{\lambda_k} \right], \quad \forall k	
\end{align}	
Similarly, under {an admissible} control policy $\boldsymbol{\Omega}$,  the average power cost of the $k$-th D2D transmitter is given by
\begin{align}		\label{power_cost}
	\overline{P}_k(\boldsymbol{\Omega}) = \limsup_{T \rightarrow \infty} \frac{1}{T} \sum_{t=0}^{T-1} \mathbb{E}^{\boldsymbol{\Omega}}  \big[P_k(t)  \big], \quad \forall k
\end{align}

We formulate the dynamic power control problem for the delay-aware D2D system as follows:
\begin{problem}	[Power Control for Delay-Aware D2D Systems]  \label{IHAC_MDP}
The power control problem for the delay-aware D2D communications is formulated as
\begin{align}\label{formula}
\underset{\boldsymbol{\Omega}}{\min}~~~&L(\boldsymbol{\Omega})\\
=& \sum_{k=1}^K\Big(\beta_k\underbrace{\overline{D}_k(\boldsymbol{\Omega})}\limits_{\text{average} \atop \text{delay}} +  \gamma_k\underbrace{\overline{P}(\boldsymbol{\Omega})}\limits_{\text{average} \atop \text{power}}\Big)\nonumber\\
=& \limsup_{T \rightarrow \infty} \frac{1}{T} \sum_{t=0}^{T-1} \mathbb{E}^{\boldsymbol{\Omega}}  \left[c\left(\mathbf{Q}\left(t\right), \boldsymbol{\Omega}\left(\boldsymbol{\chi}\left(t\right) \right)\right)    \right]	\nonumber
\end{align}
where $c\left(\mathbf{Q}, \mathbf{P}\right)=\sum_{k=1}^K \left(\beta_k \frac{Q_k}{\lambda_k}+\gamma_k P_k \right)$. $\mbox{\boldmath$\beta$}=\{\beta_k>0: \forall k\}$ and $\mbox{\boldmath$\gamma$}=\{\gamma_k>0: \forall k\}$ are positive weights for the delay cost and the power cost respectively. ~\hfill\IEEEQED
%\footnote{This is equivalent to the problem in which one of the costs is optimized with the constraint on the other cost by treating the Lagrangian multipliers as the associated weights.}
\end{problem}

{Problem \ref{IHAC_MDP} embraces various optimization formulations such as minimizing the average delay subject to the average power constraint or minimizing the average transmit power subject to the average delay constraint. This is because these ``constrained optimization problems" have the same {\em Lagrangian function}, which is given by (\ref{formula}) in Problem \ref{IHAC_MDP}. The weights $\boldsymbol{\beta}$ and $\boldsymbol{ \gamma}$ are equivalent to the Lagrangian multipliers of the associated constraints.}
%For any given weight vector $\boldsymbol{\beta}$ and $\boldsymbol{ \gamma}$, the solution of Problem \ref{IHAC_MDP} corresponds to a Pareto optimal tradeoff  between the average delay costs and the average power costs of the $K$ D2D flows.
Also note that Problem \ref{IHAC_MDP} is an infinite horizon average cost MDP, which is known as a very difficult problem.

\subsection{Optimality Conditions for Power Control Problem}
Problem \ref{IHAC_MDP} is an MDP, and the associated  {\em Bellman equation} \cite{DPcontrol}  involves the entire system state $\boldsymbol{\chi}=(\boldsymbol{\sigma},\mathbf{H},\mathbf{Q})$.
%To simplify the problem, we exploit the i.i.d. properties of $\mathbf{H}(t)$ and $\boldsymbol{\sigma}(t)$ and derive an {\em equivalent Bellman equation} which involves $\mathbf{Q}$ only. We first define the partitioned actions as below.
%\begin{definition}[Partitioned Actions]\label{d_part}
%Given a control policy $\boldsymbol{\Omega}$, we define $\boldsymbol{\Omega}(\mathbf{Q})=\{\Omega_k(\mathbf{Q}):\forall k\}$, where $\Omega_k(\mathbf{Q})=\{P_k=\Omega_k(\boldsymbol{\sigma},\mathbf{H},\mathbf{Q}):\forall \boldsymbol{\sigma},\mathbf{H}\}$ is the collection of the power control actions of transmitter $k$ for all possible CSI $\mathbf{H}$ and MAC output $\boldsymbol{\sigma}$ conditioned on a given QSI $\mathbf{Q}$. The complete control policy is therefore equal to the union of all partitioned actions, i.e., $\boldsymbol{\Omega}=\bigcup_\mathbf{Q}\boldsymbol{\Omega}(\mathbf{Q})$.\hfill \IEEEQED
%\end{definition}
Exploiting the i.i.d. properties of $\mathbf{H}(t)$ and $\boldsymbol{\sigma}(t)$, we obtain the following {\em equivalent Bellman equation}.

\begin{theorem}[Sufficient Conditions for Optimality]\label{theorem1}
For any given weights $\mbox{\boldmath$\beta$}$ and $\mbox{\boldmath$\gamma$}$, assume there exists a $(\theta^*,\{V^*(\mathbf{Q})\})$ that solves the following {\em equivalent Bellman equation}:
\begin{align}\label{bellman1}
	&\theta^*\tau+V^*(\mathbf{Q})  \hspace{4cm} 	\forall \mathbf{Q}\in \boldsymbol{\mathcal{Q}}\\
=&\mathbb{E}\bigg[ \min_{\boldsymbol{\Omega}(\boldsymbol{\chi})}\Big[c\big(\mathbf{Q}, \boldsymbol{\Omega}\big(\boldsymbol{\chi}\big)\big) \tau +\sum_{\mathbf{Q}'}\Pr \big[ \mathbf{Q}'\big| \boldsymbol{\chi}, \boldsymbol{\Omega}\big(\boldsymbol{\chi} \big) \big]V^*(\mathbf{Q}')  \Big] \bigg| \mathbf{Q} \bigg]\nonumber
\end{align}
Furthermore, for all admissible control policy $\boldsymbol{\Omega}$, $V^\ast$ satisfies the following \emph{transversality condition}:
\begin{equation}\label{trans1}
\lim_{T \rightarrow \infty} \frac{1}{T}\mathbb{E}^{\boldsymbol{\Omega}}\left[ V^\ast\left(\mathbf{Q}\left(T \right) \right)\right]=0
\end{equation}
Then ${\theta^*}=\underset{\boldsymbol{\Omega}}{\min} L(\boldsymbol{\Omega})$ is the optimal average cost, and $V^\ast\left(\mathbf{Q}\right)$ is the \emph{priority function} of the $K$ data flows. If $\boldsymbol{\Omega}^{\ast}\left(\boldsymbol{\chi}\right)$ attains the minimum of the R.H.S. of (\ref{bellman1}) for all $\mathbf{Q} \in \boldsymbol{\mathcal{Q}} $, then $\boldsymbol{\Omega}^*$ is the optimal control policy for Problem~\ref{IHAC_MDP}.
\hfill \IEEEQED
\end{theorem}

\begin{proof}
Please refer to Appendix A.
\end{proof}
\begin{remark}	[Interpretation of {Theorem} \ref{theorem1}]
%The equivalent Bellman equation in (\ref{bellman1}) is defined on the QSI $\mathbf{Q}$ only. Nevertheless, the optimal control policy $\boldsymbol{\Omega}^{\ast}$ obtained by solving (\ref{bellman1}) is still adaptive to the entire system state $\boldsymbol{\chi}$.
{At each stage when the queue length is $\mathbf{Q}(t)$, the optimal action has to strike a balance between the current cost and the future cost because the action taken will affect the future evolution of $\mathbf{Q}(t+1)$.}
Furthermore, based on the unichain property of the admission control policy, the solution obtained from Theorem \ref{theorem1} is unique \cite{DPcontrol}.~\hfill\IEEEQED
\end{remark}

\section{Low-Complexity Power Control Solution}
One key obstacle in deriving the optimal power control policy $\boldsymbol{\Omega}^\ast$ is to obtain the priority function for the Bellman equation in (\ref{bellman1}).  Conventional brute force value iteration or policy iteration algorithms can only give numerical solutions and have exponential complexity in $K$, which is highly undesirable. In this section, we shall exploit the interference filtering property of the MAC protocol and adopt perturbation theory to obtain a closed-form approximation of the priority function $V^\ast (\mathbf{Q})$ and derive the associated error bound. Based on that, we obtain a low complexity dynamic power control algorithm for the delay-aware D2D communications.

\subsection{Closed-Form Approximate Priority Function via Perturbation Analysis}
We adopt a calculus approach to obtain a closed-form approximate priority function. We first have the following theorem for solving the Bellman equation in (\ref{bellman1}).
{
\begin{theorem}	[Calculus Approach for Solving (\ref{bellman1})]	\label{HJB1}
	Assume there exist $c^\infty$ and  $J\left( \mathbf{Q};L^\delta \right)$ of class $\mathcal{C}^2(\mathbb{R}_+^K)$   that satisfy
	\begin{itemize}
		\item the following partial differential equation (PDE):
		\begin{align}	
	& \mathbb{E}\Bigg[ \min_{\boldsymbol{\Omega}\left(\boldsymbol{\chi} \right)}\bigg[ \sum_{k=1}^K \left(\beta_k \frac{Q_k}{\lambda_k}+\gamma_k P_k \right)- c^{\infty} \hspace{0.8cm} \forall \mathbf{Q} \in \mathbb{R}_+^K \nonumber\\
&+ \sum_{k=1}^K \bigg( \frac{\partial  J \left(\mathbf{Q};L^\delta \right)} {\partial Q_k} \left(\lambda_k - \sigma_kC_k\big(\mathbf{H},\mathbf{P} \big) \right) \bigg)\bigg]\Bigg| \mathbf{Q} \Bigg] =0	\label{bellman3}
\end{align}
with boundary condition $J\left(\mathbf{0};L^\delta \right)=0$.
        \item $\Big\{\frac{\partial J\left(\mathbf{Q}; L^\delta \right)}{\partial Q_k}: \forall k \Big\}$  are increasing functions of all $Q_k$.
		\item $J\left(\mathbf{Q}; L^\delta \right)=\mathcal{O}\left(\|\mathbf{Q}\|^3\right)$.
	\end{itemize}
	Then, we have
\begin{equation}	
\theta^\ast= c^{\infty}+o(1), V^\ast\left(\mathbf{Q} \right)=J \left(\mathbf{Q};L^\delta  \right)+o(1), \forall \mathbf{Q} \in \boldsymbol{\mathcal{Q}}	\label{15resu}
\end{equation}
where the error term $o(1)$ asymptotically goes to zero  for sufficiently small  $\tau$.~\hfill\IEEEQED
\end{theorem}}

\begin{proof}
	please refer to Appendix B.
\end{proof}

{Theorem \ref{HJB1} suggests that if we can solve for the PDE in (\ref{bellman3}), then the solution $(J \left(\mathbf{Q};L^\delta  \right), c^{\infty})$ is only $o(1)$ away from the solution of the Bellman equation $(V^*(\mathbf{Q}), \theta^*)$.} Before we solve the $K$-dimensional PDE in (\ref{bellman3}), we first recognize that due to the interference filtering property of the MAC protocol in Assumption \ref{a_MAC}, the cross-channel path gain of all the active D2D flows  are quite weak and the worst-case interfering path gain is $L^\delta $. Note that the solution of  (\ref{bellman3})  depends on the worst-case cross-channel path gain $L^\delta$ and, hence, the $K$-dimensional PDE in (\ref{bellman3}) can be regarded as a perturbation of a {\em base system} defined below.
\begin{definition}[Base System]\label{base}
A base system is characterized by the PDE in (\ref{bellman3}) with $L^\delta=0$.~\hfill \IEEEQED
\end{definition}

We then study the base system and use  $J(\mathbf{Q};0)$ to obtain a closed-form approximation of $J(\mathbf{Q};L^\delta)$.
%We have the following lemma regarding the MAC protocol in Assumption \ref{a_MAC}  when $\delta=\infty$.
%\begin{lemma}   [TDMA Property of MAC Protocol when $\delta=\infty$]    \label{infpri}
%    When the sensing distance $\delta=\infty$,  the MAC protocol in Assumption \ref{a_MAC} corresponds to a TDMA channel access mode. Specifically, at each time slot $t$, only one out of $K$ transmit node is allowed to transmit and each transmit node has equal opportunity to access the channel with probability $v_k\left(\infty \right)=\frac{1}{K}$ ($\forall k$).~\hfill \IEEEQED
%\end{lemma}
%\begin{proof}
%	please refer to Appendix D.
%\end{proof}
%
%As a result of Lemma \ref{infpri}, the $K$ queue dynamics in the base system are decoupled due to the absence of interference using \emph{TDMA-like} MAC protocol.
We have the following lemma summarizing the priority function $J(\mathbf{Q};0)$  of the base system.
\begin{lemma}	[Decomposable Structure of $J(\mathbf{Q};0)$]\label{theorem3}
The solution  $J(\mathbf{Q};0)$ for the base system has the following decomposable structure:
\begin{equation}   \label{linearA}
J\left(\mathbf{Q};0 \right)= \sum_{k=1}^K J_k \left(Q_k \right)
\end{equation}
where  $J_k \left(Q_k \right)$ is the \emph{per-flow priority function} for the $k$-th data flow given by
\begin{equation}	\label{perflow}
 \left\{
\begin{aligned}	
Q_k(y)=&  \frac{\lambda_k}{\beta_k} \left(\frac{a_k}{(|\mathcal{N}_k(\delta)|+1) \ln 2}E_1\left(\frac{a_k}{y}\right)- \lambda_ky \right. \\
& \left. -\frac{y}{(|\mathcal{N}_k(\delta)|+1) \ln 2}\left(e^{-\frac{a_k}{y}}-E_1\left(\frac{a_k}{y}\right)\right)+c_k^\infty\right)		 \\
J_k(y)=&   \frac{\lambda_k}{\beta_k} \left(\frac{1}{4 (|\mathcal{N}_k(\delta)|+1) \ln 2}E_1\left(\frac{a_k}{y}\right)\left(2y^2-a_k^2\right) \right.  \\
& \left. -\frac{y(y-a_k)}{4 (|\mathcal{N}_k(\delta)|+1) \ln 2}e^{-\frac{a_k}{y}}-\frac{\lambda_ky^2}{2}\right)+b_k
\end{aligned}
\right.
\end{equation}
where $a_k \triangleq \frac{N_0{\Gamma}\gamma_k\ln 2}{L_{kk}}$. $c_k^\infty=\frac{1}{(|\mathcal{N}_k(\delta)|+1) \ln 2}\left(d_k e^{-\frac{a_k}{d_k}} -a_k E_1\left(\frac{a_k}{d_k}\right)\right)$, where $d_k$ satisfies $\frac{1}{ (|\mathcal{N}_k(\delta)|+1) \ln 2}E_1\left(\frac{a_k}{d_k}\right)=\lambda_k$. $E_1(z)  \triangleq \int_1^{\infty} \frac{e^{-tz}}{t}\mathrm{d}t = \int_z^{\infty} \frac{e^{-t}}{t}\mathrm{d}t $.  $b_k$ is chosen to satisfy\footnote{To find $b_k$, firstly solve $Q_k(y_k^0)=0$ using one-dimensional search techniques (e.g., bisection method). Then $b_k$ is chosen such that $J_k(y_k^0)=0$.} the boundary condition $J_k(0)=0$.~\hfill\IEEEQED
\end{lemma}
\begin{proof}
please refer to Appendix C.
\end{proof}

Note that when $L^\delta = 0$, the interference network has $L_{kj} = 0$ for all $k \neq j$ with $d_{kj}> \delta$ and, hence, there is no interference between the active D2D ndoes. As a result, the $K$ D2D flows are totally decoupled and the system is equivalent to a decoupled system with $K$ independent D2D flows. That is  why the priority function $J\left(\mathbf{Q};0 \right)$ in the base system has the decomposable structure in Lemma \ref{theorem3}.

We then analyze the asymptotic property of the per-flow priority function $J_k \left( Q_k \right)$ in Corollary~\ref{cor2}.
\begin{corollary}	[Asymptotic Property of $J_k \left( Q_k \right)$]   	\label{cor2}
\begin{align}	\label{asympotic}
J_k\left(Q_k \right) &= \frac{ \beta_k (|\mathcal{N}_k(\delta)|+1)}{2\lambda_k}\frac{Q_k^2}{\log_2 \left( Q_k \right)} +o\left(\frac{Q_k^2}{\log_2 \left( Q_k \right)} \right),\nonumber\\
&\hspace{5cm} \text{as } Q_k \rightarrow \infty
\end{align}
\end{corollary}
\begin{proof}
Please refer to Appendix D.
\end{proof}

Next, we study the PDE in (\ref{bellman3}) for large $\delta$. Note that large $\delta$ corresponds to small cross-channel path gains within the set of active D2D nodes. Hence, $J(\mathbf{Q};L^\delta)$ can be considered as a perturbation of the solution of the base system $J(\mathbf{Q};0)$. Using perturbation analysis, we establish the following theorem on the approximation of  $J(\mathbf{Q};L^\delta)$:
\begin{theorem}	[First Order Approximation of $J\left(\mathbf{Q}; L^\delta \right)$]	\label{theorem4}
	 $J\left(\mathbf{Q}; L^\delta \right)$ can be approximated by $J\left(\mathbf{Q}; 0 \right)$,  and the first order perturbation term is given by
\begin{align}	\label{Jerror}
    J\left(\mathbf{Q}; L^\delta \right)=& J\left(\mathbf{Q}; 0 \right) + \sum_{k=1}^K \sum_{j \neq k \atop j \notin \mathcal{N}_k(\delta)}   \Big(\frac{D_{kj}L_{kj}Q_k^2 Q_j}{(\log_2 Q_k)^2 \log_2 Q_j}\nonumber\\
&+o\left(\frac{D_{kj}L_{kj}Q_k^2 Q_j}{(\log_2 Q_k)^2 \log_2 Q_j}\right)   \Big)+ \mathcal{O}\left(\frac{1}{\delta^4} \right)
\end{align}
where $D_{kj}=\frac{\beta_k\beta_j (|\mathcal{N}_k(\delta)|+1) }{2 (\ln2)   \lambda_k \lambda_j \gamma_j N_0} $.~\hfill\IEEEQED
\end{theorem}
\begin{proof}
Please refer to Appendix E.		
\end{proof}

{The priority function $V(\mathbf{Q})$ is decomposed into the following three terms: 1) the base term $\sum_k J_k(Q_k)$ obtained by solving a base system without coupling, 2) the perturbation term accounting for the first order {\em interference coupling} due to simultaneously transmitting D2D nodes after MAC filtering, and 3) the residual error term. As a result, we adopt the following closed-form approximation of $V(\mathbf{Q})$:}
\begin{equation}	\label{finalapprox1}
	%V\left( \mathbf{Q} \right) \approx
\widetilde{V}\left( \mathbf{Q} \right) \triangleq  \sum_{k=1}^K J_k \left( Q_k \right)+\sum_{k=1}^K \sum_{j \neq k \atop j \notin \mathcal{N}_k(\delta)}   \frac{D_{kj}L_{kj}Q_k^2 Q_j}{(\log_2 Q_k)^2  \log_2 Q_j}
\end{equation}

\begin{remark}  [Approximation Error w.r.t. System Parameters]\
\vspace{-0.5cm}
    \begin{itemize}
        \item \textbf{Approximation Error w.r.t. Traffic Loading:} the approximation error is a decreasing function of the average arrival rate $\lambda_k$.
        \item \textbf{Approximation Error w.r.t. SNR:} the approximation error is an increasing function of the SNR (which is a decreasing function of $\gamma_k$).
        \item  \textbf{Approximation Error w.r.t. Sensing Distance:} the approximation error is a decreasing function of the carrier sensing distance at the order\footnote{For any $k$, $j \neq k$ and $j \notin \mathcal{N}_k(\delta)$, we have $d_{kj}>\delta$. Therefore, according to the long term path gain model in Assumption \ref{long}, we have $L_{kj}=\frac{G_k^r G_j^t \lambda^2}{\left( 4 \pi d_{kj}\right)^2} = \mathcal{O}\left( \frac{1}{\delta^2}\right)$.} at least  $\mathcal{O}\left(\frac{1}{\delta^2}\right)$.
    \end{itemize}
\end{remark}

%\begin{figure}
%  \centering
%  \includegraphics[width=3in]{grid}
%\caption{{Structure of the priority function $\widetilde{V}(\mathbf{Q})$.}}
%  \label{fig_priority}	\vspace{-0.2cm}
%\end{figure}
%
%
%{Fig.~\ref{fig_priority} illustrates the structure of $\widetilde{V}(\mathbf{Q})$ for $K=2$ to provide some insights.
%The priority weights of flow 1 and flow 2 ($\frac{\partial \widetilde{V}(\mathbf{Q})}{\partial Q_1}$ and $\frac{\partial \widetilde{V}(\mathbf{Q})}{\partial Q_2}$) are the increasing functions of $(Q_1,Q_2)$.  Furthermore, they are dominated by the longer queue\footnote{For example, if $Q_1>Q_2$, the slope of $\widetilde{V}(\mathbf{Q})$ w.r.t. $Q_1$ is larger than that w.r.t. $Q_2$, i.e., $\frac{\partial \widetilde{V}(\mathbf{Q})}{\partial Q_1}>\frac{\partial \widetilde{V}(\mathbf{Q})}{\partial Q_2}$, which implies that the change of $Q_1$ has a larger effect on the priority function $\widetilde{V}(\mathbf{Q})$.}, as illustrated in Fig.~\ref{fig_priority}.
From Corollary~\ref{cor2} and (\ref{finalapprox1}), the {\em priority function} $\widetilde{V}(\mathbf{Q}) = \mathcal{O}(\frac{Q_k^2}{\log Q_k})$ for large $Q_k,\forall k$. As a result, the longer queue will get higher priority in the order of $\frac{Q_k}{\log Q_k}$. %When $Q_1$ is larger, the priority function $V(\mathbf{Q})$ increases more quickly with the increase of $Q_1$, i.e., $\frac{\partial V(\mathbf{Q})}{\partial Q_1}$ is larger, which indicates that $V(\mathbf{Q})$ could increase the effect of QSI for long queues and reflect the urgency of these flows appropriately. More interestingly, when $Q_2$ is larger, $\frac{\partial V(\mathbf{Q})}{\partial Q_1}$ is also larger. In that case, $\frac{\partial V(\mathbf{Q})}{\partial Q_2}$ is large because of its long queue, so $Q_1$ should adjust its priority more quickly for achieving fairness between flows.}
%Furthermore, based on Corollary \ref{cor2}  and (\ref{theorem4}), $\left\{\frac{\partial \widetilde{V}\left(\mathbf{Q}\right)}{\partial Q_k}: \forall k \right\}$  are increasing functions of all $Q_k$ and $\widetilde{V}\left(\mathbf{Q}\right)=\mathcal{O}\left(\|\mathbf{Q}\|^3\right)$.
Based on Theorem \ref{theorem3} and Theorem \ref{theorem4}, the approximation error between the optimal priority function $V^\ast \left(\mathbf{Q} \right)$ in Theorem \ref{theorem1} and the closed-form approximate priority function  $\widetilde {V}\left(\mathbf{Q} \right)$ in (\ref{finalapprox1})  is $\mathcal{O}(\frac{1}{\delta^2})+ o(1)$. In other words, the error terms are asymptotically small w.r.t. the carrier sensing distance $\delta$ and the slot duration.
%In the next section, we  derive a low complexity power control policy using the closed-form approximate priority function $\widetilde {V}\left(\mathbf{Q} \right)$ in (\ref{finalapprox1}).

\subsection{Asymptotically Delay-Optimal  Power Control Algorithm}
In this section, we use the closed-form approximate priority function in (\ref{finalapprox1}) to capture the urgency information of the $K$ D2D pairs and  obtain low complexity delay-aware power control. Using the approximate priority function in (\ref{finalapprox1}) and Lemma \ref{cor1}, the per-stage control problem (for each state realization $\boldsymbol{\chi}$) is given by\footnote{Note that $J_k'\left(Q_k \right) = \left(\frac{\mathrm{d} J_k\left(y\right)}{\mathrm{d} y } \Big/\frac{\mathrm{d} Q_k\left(y\right)}{\mathrm{d} y }\right)\Big|_{y=y\left(Q_k \right)} = y\left(Q_k \right)$, where $ y\left(Q_k \right)$ satisfies $Q_k\left(  y\left(Q_k \right) \right)=Q_k$.}
{
\begin{align}	\label{utility}
\max _{\mathbf{P} }  \ \sum_{k=1}^K  \Big(\underbrace{\frac{\partial \widetilde{V}\left(\mathbf{Q}\right)}{\partial Q_k}}_{\text{flow weight}}  \sigma_k \underbrace{ C_k\left(\mathbf{H}, \mathbf{P}\right)}_{\text{data rate}} - \gamma_k  P_k \Big)
\end{align}
}
where $\frac{\partial \widetilde {V}\left(\mathbf{Q}\right)}{\partial Q_k}$ can be calculated from (\ref{finalapprox1}) which is given by
\begin{align}
&\frac{\partial \widetilde {V}\left(\mathbf{Q}\right)}{\partial Q_k} = J_k'\left(Q_k\right)\\
&+ \sum_{j \neq k\atop j \notin \mathcal{N}_k(\delta)} \frac{Q_j (\ln Q_k -1)}{(\ln 2)  \log_2^2 Q_k \log_2 Q_j} \left( \frac{2 D_{kj}L_{kj} Q_k}{\log_2 Q_k} +  \frac{D_{jk}L_{jk} Q_j}{\log_2 Q_j}\right)\nonumber
\end{align}

%\vspace{-0.5cm}
{%Note that the problem in (\ref{utility}) can be treated as a weighted sum-rate problem and the associated  power control algorithm can be obtained using an iterative algorithm \cite{weiyu}. Unlike a traditional weighted sum-rate problem,
The per-stage problem in (\ref{utility}) is similar to the weighted sum-rate (WSR) optimization subject to the power constraint, which has been widely studied in \cite{weiyu} and \cite{WSR}. However, unlike conventional WSR problems where the weights are static, the weights here in (\ref{utility}) are dynamic and are determined by the QSI via the priority function $\frac{\partial \widetilde{V}(\mathbf{Q})}{\partial Q_k}$. As such, the role of the QSI is to dynamically adjust the weight (priority) of the individual flows, whereas the role of the CSI is to adjust the priority of the flow based on the transmission opportunity in the rate function $C_k(\mathbf{H},\mathbf{P})$.} Note that the per-stage problem in (\ref{utility}) is challenging due to the non-convexity of $C_k\left(\mathbf{H}, \mathbf{P}\right)$ w.r.t. $\mathbf{P}$. We shall first derive a low complexity iterative solution that converges to the stationary point of (\ref{utility}). We then show that the converged solution is asymptotically optimal for sufficiently small $L^\delta$.
\begin{algorithm}[Delay-Aware Dynamic Power Control]\label{algorithm}\
\begin{itemize}	
\item \textbf{Step 1 [Initialization]:} Let $n=0$. Initialize a feasible $\mathbf{P}(0)$.

\item \textbf{Step 2 [Iteration]:} In the $(n+1)$-th iteration, the transmit power of each D2D transmitter is updated  based on the power results of the $n$-th iteration according to
    \begin{align}
    P_k(n+1)= \left(\frac{\partial \widetilde {V}\left(\mathbf{Q}\right)}{ \partial Q_k}\frac{1}{(\ln 2) (\gamma_k + \zeta_k(n))}-\frac{{\Gamma} I_k(n)}{{H}_{kk}}\right)^+
    \end{align}
    where $I_k(n)=N_0+\sum_{j \neq k, j \notin \mathcal{N}_k(\delta)} {H}_{kj} P_j(n)$ and $\zeta_k(n)=\sum_{j \neq k, j \notin \mathcal{N}_k(\delta)} \frac{1}{\ln 2}\frac{\partial \widetilde {V}\left({Q}\right)}{ \partial Q_j} \frac{{H}_{jj}P_j(n){H}_{kj}}{ I_j(n) (I_j(n)+{H}_{jj}P_j(n)) }$.

\item \textbf{Step 3 [Termination]:} Set $n = n + 1$ and go to Step 2 until a certain termination condition is satisfied.~\hfill\IEEEQED
\end{itemize}	
\end{algorithm}

%\begin{remark}  [Convergence Property of Algorithm \ref{algorithm}] \label{lemmaconv}
%    The convergence for Algorithm \ref{algorithm} can be proved using the contraction mapping  approach \cite{cont1}. Specifically, any limiting point $\mathbf{P}(\infty)$ is a stationary point of the problem in (\ref{utility}).~\hfill\IEEEQED
%\end{remark}
%
%
%Although the problem in (\ref{utility}) is not convex in general, the following corollary states that the limiting point $\mathbf{P}\left( \infty\right)$ of Algorithm \ref{algorithm}  is asymptotically optimal for sufficiently small $L^\delta$.

Although the problem in (\ref{utility}) is non-convex in general, we show below that Algorithm \ref{algorithm} converges to the global optimal solution asymptotically for sufficiently large $\delta$.

\begin{corollary}  [Asymptotic Optimality of Algorithm \ref{algorithm}]  	\label{collaryAlg}
Algorithm \ref{algorithm} converges to the unique global optimal point of  the problem in (\ref{utility}) for sufficiently large $\delta$.~\hfill\IEEEQED
\end{corollary}
\begin{proof}
Please refer to Appendix F.
\end{proof}

\subsection{Summary of the Overall Solution and Implementation Considerations}
We give a summary of the overall dynamic power control solution and discuss some implementation considerations (computational complexity) in the context of LTE-Advanced systems \cite{3gpp}. Specifically, we consider the scenario of {\em fully controlled D2D communications}~\cite{D2D7} in LTE-Advanced in which the eNodeB takes control of the radio resource for the D2D nodes inside its coverage. A frame is divided into a {\em contention phase}, a {\em reporting phase}, a {\em decision phase} and a {\em data transmission phase}, which are described as follows:
\begin{enumerate}
\item \textbf{Contention Phase:} D2D nodes access the channels distributively according to a CSMA-like MAC protocol. At the end of the contention phase, each D2D transmitter gets its corresponding MAC output $\sigma_k(t)$. Also, during this phase, the CSI $\mathbf{H}(t)$ could be estimated by the D2D receivers\footnote{Each active D2D transmitter has to send the control signaling for MAC contention. The CSI can be estimated if the signaling is sent with a given power, i.e., as the reference signal.}.

\item \textbf{Reporting Phase:} Each of the active transmitters ($\mathcal{A}(t)=\{ k: \sigma_k(t)=1  \}$) report their local CSI $\{H_{kj}(t): \forall j\}$ and local QSI $Q_k(t)$ to the eNodeB via Physical Uplink Control Channel (PUCCH) and Physical Uplink Shared Channel (PUSCH) \cite{3gpp4}, respectively.

\item \textbf{Decision Phase:} After receiving the CSI and QSI reports, the eNodeB calculates the optimal power for the active D2D nodes according to the proposed Algorithm \ref{algorithm}, and broadcasts the power control actions  to the  active D2D nodes via Physical Downlink Control Channel (PDCCH) \cite{3gpp4}.

\item \textbf{Data Transfer Phase:} The active D2D transmitters adjust their transmit power according to the  power control broadcasted from the eNodeB and transmit data during the data transmission phase in the current frame.
\end{enumerate}

%\begin{figure}
%  \centering
%  \includegraphics[width=3.5in]{frame.eps}\\
%  \caption{Frame structure for algorithm implementation}
%  \label{fig_frame}\vspace{-0.5cm}
%\end{figure}

\begin{remark}	[Computational Complexity Consideration]
The computational complexity of the proposed solution is very low. Specifically, most of complexity comes from computing the priority function in (\ref{finalapprox1}) and computing the power control actions using Algorithm \ref{algorithm}. The complexity of computing the priority function is very low (due to the closed form characterization) compared with conventional brute-force value iterations algorithms \cite{DPcontrol}, which have exponential complexity in $K$. Computing the power control actions using Algorithm \ref{algorithm} is similar to those conventional iterative water-filling solutions for solving WSR optimization in \cite{weiyu}. We shall quantify the complexity comparison in Section V.~\hfill\IEEEQED %Table~\ref{tabletime} illustrates the comparison of the MATLAB computational time of the proposed solution, the baselines and the brute-force value iteration algorithm \cite{DPcontrol}.
\end{remark}

\begin{remark}	[Extension for OFDMA and General Fading]
The solution framework in Theorem 2 and Theorem 3 can be extended easily to multi-channel systems (such as OFDMA \cite{OFDMA2}) as well as general fading distributions. For OFDMA systems, the modification required is the rate equation in (\ref{capacity}). Each channel can be treated independently since orthogonal parallel channels do not introduce additional coupling. For general fading distributions, the modification required is the solution of the per-flow PDE in the base system $J_k \left(Q_k \right)$ in Lemma~\ref{theorem3}.
%In multi-channel systems, e.g., LTE-Advanced systems based on OFDMA \cite{OFDMA1,OFDMA2}, our solution framework can be implemented independently for each channel because of the orthogonal parallel channels without introducing additional coupling.
~\hfill\IEEEQED
\end{remark}

%\begin{figure*}
%\centering
%\begin{minipage}[t]{0.45\textwidth}
%  \centering
%  \includegraphics[width=2.5in]{sim1.eps}
%  \caption{Performance comparison with different average arrival rates}
%  \label{sim1}
%\end{minipage}
%\hspace{0.05\textwidth}
%\begin{minipage}[t]{0.45\textwidth}
%  \centering
%  \includegraphics[width=2.5in]{sim2.eps}
%  \caption{Performance comparison with different average transmit power}
%  \label{sim2}
%\end{minipage}	%\vspace{-0.5cm}
%\end{figure*}
%
%
%\begin{figure*}
%\centering
%\begin{minipage}[t]{0.45\textwidth}
%  \centering
%  \includegraphics[width=2.5in]{sim3.eps}
%  \caption{Performance comparison with different D2D communication ranges}
%  \label{sim3}
%\end{minipage}
%\hspace{0.05\textwidth}
%\begin{minipage}[t]{0.45\textwidth}
%  \centering
%  \includegraphics[width=2.5in]{sim4.eps}
%  \caption{Effect of carrier sensing distance}
%  \label{sim4}
%\end{minipage}%\vspace{-1cm}
%\end{figure*}

\section{Simulation Results}    \label{sim}
In this section, we evaluate the performance of the proposed low-complexity power control scheme for D2D communications. The following four baseline schemes are adopted for performance comparison.

\begin{itemize}
\item \textbf{Baseline 1 [Cellular Mode]}: The Tx-Rx pairs transmit their data via the cellular BS in a conventional way \cite{D2D4}. The $K$ pairs share the channel using TDMA in a Round-Robin way.

\item \textbf{Baseline 2 [D2D with Fixed Power]}: The transmitters always transmit with the maximum power for D2D communications \cite{D2D4}.

\item \textbf{Baseline 3 [D2D with CSI-based Power Control]}: {Large deviation \cite{comparison1} is an approach to bypass the complex delay minimization by converting the delay constraint into an equivalent rate constraint.} The CSI-based power control scheme determines the transmit power for maximizing the total data rate without considering the queueing information \cite{D2D8}.

\item \textbf{Baseline 4 [D2D with Queue-weighted Power Control]}: {Lyapunov drift approach \cite{comparison2} considers queue stabilization instead of delay minimization. The queue-weighted power control scheme exploits both CSI and QSI, and solves the per-stage problem (\ref{utility}) replacing $\frac{\partial \widetilde{V}(\mathbf{Q})}{\partial Q_k}$ with  $Q_k$.} It is similar to the Modified Largest Weighted Delay First algorithm in \cite{MLWDF} but with a modified objective function.
\end{itemize}

In the simulations, 10 D2D pairs are considered in a single cell with  radius 500m. The transmitters are located randomly in the cell and the receivers appear within the D2D communication range of their corresponding transmitters, which is set to 50m. The carrier sensing distance $\delta$ is 100m. Poisson data arrival is considered with a uniform distributed average arrival rate, which has  mean 5Mbps. The path gain is calculated as $L_{kj}=15.3+37.6\log_{10} d_{kj}$ \cite{3gpp5} with the fading coefficient distributed as $\mathcal{CN}(0,1)$. The average transmit power is 23dBm and the noise power spectrum density is -174dBm/Hz. The system bandwidth is 10MHz. The duration of the time slot is 1ms. The SINR gap $\Gamma$ is set to 1 in the simulation. The weights $\gamma_k$ are the same and $\beta_k=1$ for all $k$. For comparison, the delay performances of different schemes are evaluated with the same average transmit power by adjusting $\gamma_k$. For obtaining the average performance, we consider 100 random topologies, each of which has 1000 time slots.

Fig.~\ref{sim1} shows the average delay versus the average arrival rate. For large traffic load, the transmission via D2D communication has significant performance gain compared with the conventional cellular transmission. This is mainly because of the short distance between D2D transmitters and receivers and their efficient spatial reuse. It can also be observed that the proposed power control algorithm outperforms all the baselines, which verifies the accuracy of the priority function approximation in the proposed power control scheme. It is noticed that the delay of the proposed scheme with small arrival rate is not 0 but a small value, because the transmitters could not transmit data in all time slots.

Fig.~\ref{sim2} shows the average delay versus the average transmit power. The proposed power control scheme also achieves better performance than other baseline schemes. A larger transmit power could increase the received power of the desired signal, but, meanwhile, would cause more serious interference to other D2D pairs. Because of the two-fold effect of the transmit power, the change of the average delay performance is relatively small with adjustment of the average transmit power.

\begin{figure}
  \centering
  \includegraphics[width=2.5in]{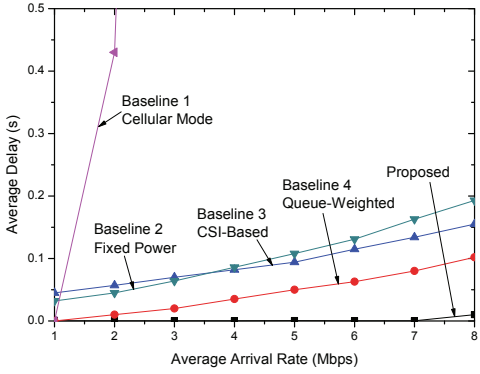}\\
  \caption{Performance comparison with different average arrival rates}
  \label{sim1}
\end{figure}

\begin{figure}
  \centering
  \includegraphics[width=2.5in]{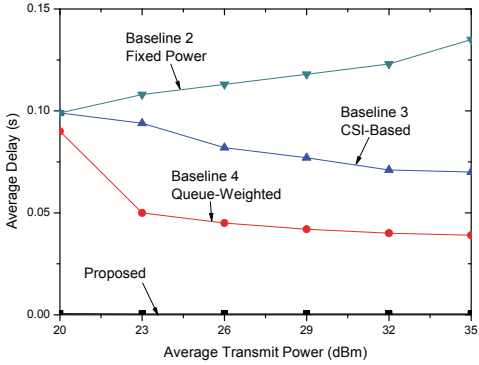}\\
  \caption{Performance comparison with different average transmit power}
  \label{sim2}	
\end{figure}

\begin{figure}
  \centering
  \includegraphics[width=2.5in]{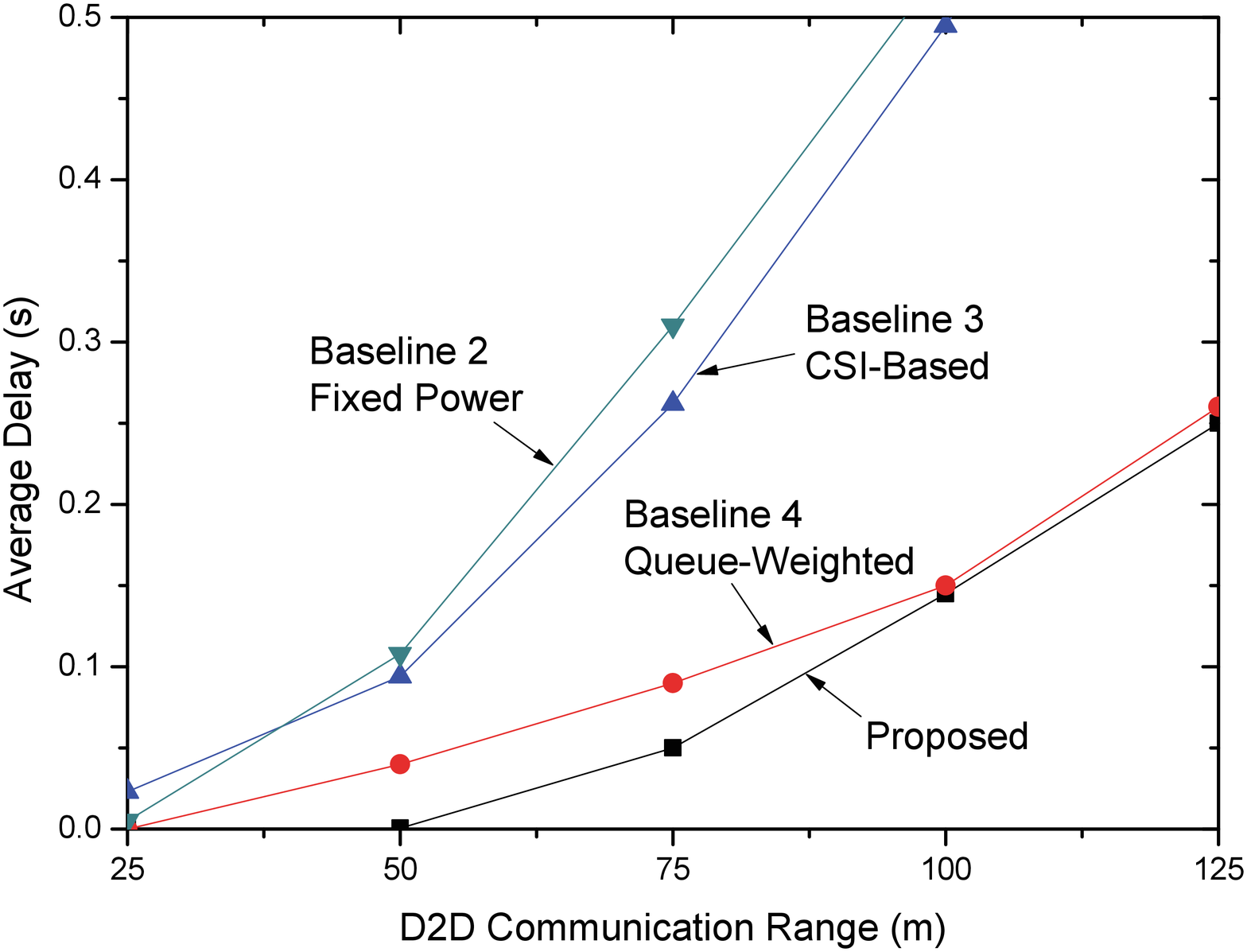}\\
  \caption{Performance comparison with different D2D communication ranges}
  \label{sim3}%\vspace{-0.3cm}
\end{figure}

\begin{figure}
  \centering
  \includegraphics[width=2.5in]{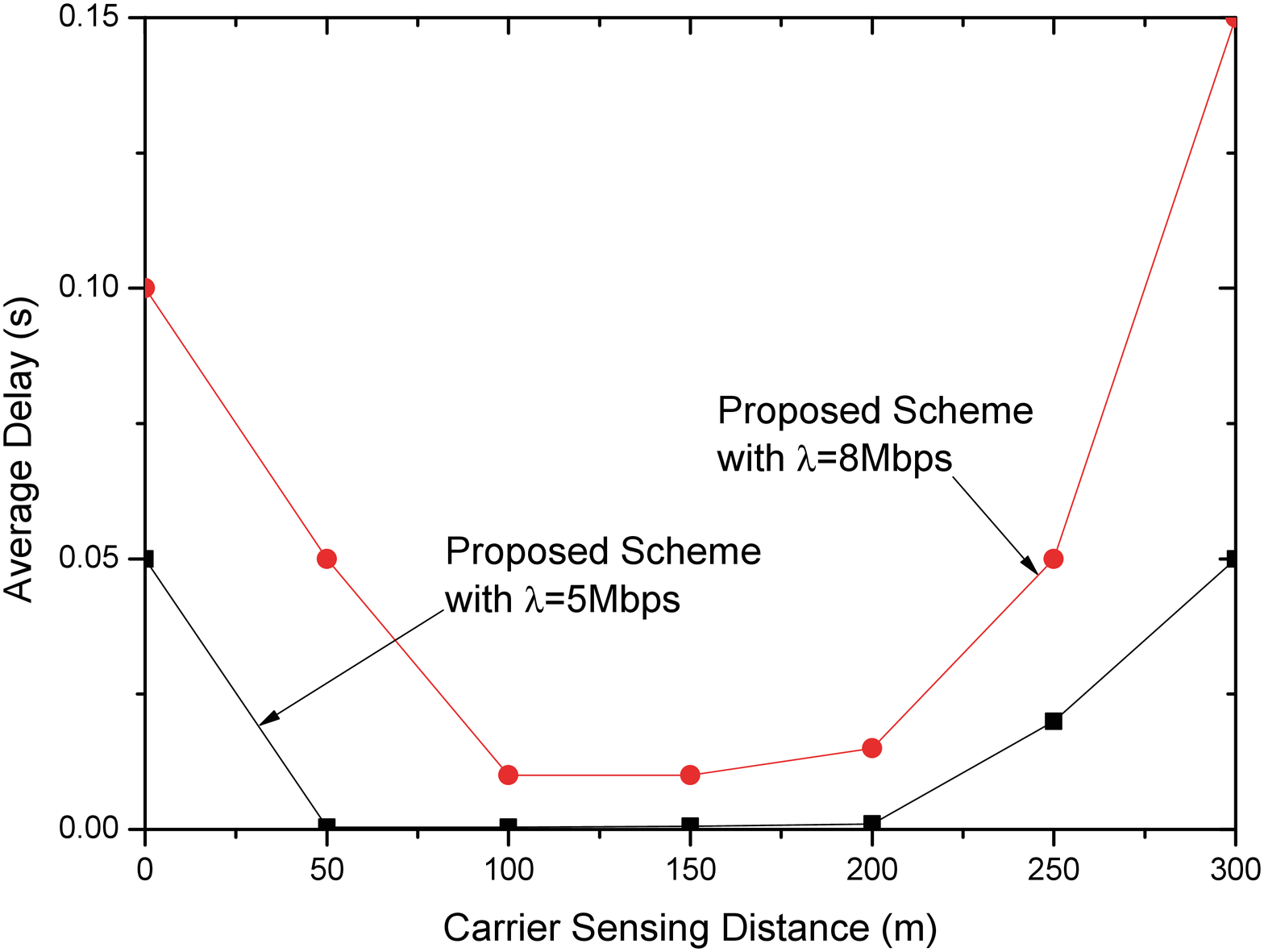}\\
  \caption{Effect of carrier sensing distance}
  \label{sim4}
\end{figure}

Fig.~\ref{sim3} indicates the average delay versus the D2D communication range. Unlike the average transmit power, the D2D communication range affects the received power of the desired signal without increasing the interference directly, so the average delay changes a lot with different D2D communication ranges. It can be found that the proposed power control scheme outperforms the baselines when the D2D communication range is small. For large D2D communication ranges (i.e., 100m and 125m), all schemes achieve quite poor delay performance. Note that since the carrier sensing distance $\delta$ is set to 100m here, MAC could not filter the large interference well. Thus, the performance of the proposed power scheme degrades because the weak coupling property of the queue dynamics does not hold when the D2D communication range is too large compared to the carrier sensing distance.

Fig.~\ref{sim4} shows the effect of carrier sensing distance $\delta$ of the proposed power control scheme. As discussed before, a very small sensing distance cannot filter the large interference or guarantee the weak coupling property of the queue dynamics. However, a very large sensing distance leads to inefficient spatial reuse. An appropriate carrier sensing distance should be selected to balance the tradeoff between the above two aspects. From Fig.~\ref{sim4}, we observe that the proposed scheme could achieve good delay performance with a large regime of carrier sensing distance.

Table \ref{tabletime}  illustrates the comparison of the MATLAB computational time of the proposed solution, the  baselines and the brute-force value iteration algorithm \cite{DPcontrol} in one time slot. Note that the computation time of Baseline 2 is the smallest in all different $K$ scenarios but it has the worst performance.  In addition,  the computational time of our proposed scheme is close to those of  Baselines 3 \& 4 and the difference  is due to the computation of the approximate priority function. Therefore, our proposed scheme achieves significant performance gain compared to all the baselines, with small computational complexity cost.

\begin{table}[!h]
\caption{Comparison of the MATLAB computational time}\label{tabletime}
\vspace{-0.3cm}
\begin{center}
\begin{tabular}{|c|c|c|c|}
\hline
& $\tcb{K=5}$  & $\tcb{K=10}$ & $\tcb{K=15}$    \\
\hline
{Baseline 2} & \tcb{$<1$ms} & \tcb{$<1$ms} & \tcb{$<1$ms}    			\\
{Baseline 3 \& 4} & \tcb{ 0.007s} & \tcb {0.015s} & \tcb{ 0.029s}\\
{Proposed Scheme} & \tcb{0.046s} & \tcb{ 0.091s}  &  \tcb{0.143s}	\\	
{Brute-Force Value Iteration} &   \tcb{$>10^5$s} & \tcb{$>10^5$s}  &  \tcb{$>10^5$s} 	\\
\hline
\end{tabular}
\end{center}
\vspace{-0.3cm}	
\end{table}

\section{Conclusion}
In this paper, we consider the dynamic power control for delay-aware D2D communications by formulating the associated stochastic optimization problem as an infinite horizon average cost MDP. To deal with the curse of dimensionality, a closed-form approximate priority function is derived using perturbation analysis. Both the analysis and the numerical results show that the approximation error is small and will vanish if the cross-channel path gain goes to 0. Based on the closed-form approximation, we propose a low complexity iterative power control algorithm and discuss some implementation issues for practical systems.  Finally, simulation results show that the proposed power control algorithm has significant performance gains in delay performance compared with various state-of-the-art baselines.

% Can use something like this to put references on a page
% by themselves when using endfloat and the captionsoff option.
%\ifCLASSOPTIONcaptionsoff
%  \newpage
%\fi

\section*{Appendix A: Proof of Theorem \ref{theorem1}}
Following \emph{Prop. 4.6.1} of \cite{DPcontrol},  the sufficient conditions for the optimality of \emph{Problem 1} are that assume  ($\theta^\ast, \{ V^\ast\left(\mathbf{Q} \right) \}$)  solves the following Bellman equation:
\begin{align}\label{belequapp1}
	&\theta^\ast {\tau} + V^\ast\left(\boldsymbol{\chi} \right)\nonumber\\
=& \min_{ \boldsymbol{\Omega} (\boldsymbol{\chi})} \Big[ c\big(\mathbf{Q}, \boldsymbol{\Omega}\big(\boldsymbol{\chi}\big)\big){\tau}+  \sum_{\boldsymbol{\chi}'} \Pr\big[ \boldsymbol{\chi}' \big| \boldsymbol{\chi}, \boldsymbol{\Omega}\big( \boldsymbol{\chi}\big)\big]  V^\ast\left(\boldsymbol{\chi}'\right)    \Big]	\nonumber \\
=& \min_{ \boldsymbol{\Omega} (\boldsymbol{\chi})}\Big[ c\big(\mathbf{Q}, \boldsymbol{\Omega}\big(\boldsymbol{\chi}\big)\big) {\tau}+ \sum_{\mathbf{Q}'} \sum_{\mathbf{H}'}\sum_{\boldsymbol{\sigma}'}  \Pr \big[ \mathbf{Q}'\big| \boldsymbol{\chi}, \boldsymbol{\Omega}\big(\boldsymbol{\chi}  \big) \big]\nonumber\\
&\hspace{3.5cm} \Pr \big[\mathbf{H}' \big]\Pr \big[\boldsymbol{\sigma}' \big]  V^\ast\left(\boldsymbol{\chi}' \right)    \Big]
\end{align}
and  $V^\ast$ satisfies the  condition in (\ref{trans1})  for all  admissible  policies $\boldsymbol{\Omega}$. Then ${\theta^*}=\underset{\boldsymbol{\Omega}}{\min} L(\boldsymbol{\Omega})$. Taking expectation w.r.t. $\mathbf{H}$ and $\boldsymbol{\sigma}$ on both sizes of (\ref{belequapp1}) and denoting $V^\ast\left(\mathbf{Q} \right) = \mathbb{E}\big[V^\ast\left(\boldsymbol{\chi}\right) \big| \mathbf{Q}\big]$, we obtain the equivalent Bellman equation in (\ref{trans1}) in Theorem \ref{theorem1}.

\section*{{Appendix B: Proof of Theorem \ref{HJB1}}}
In the proof, we shall first establish the relationship between the equivalent Bellman equation in (\ref{bellman1}) in Theorem \ref{HJB1} and the approximate Bellman equation in (\ref{bellman2}) in the following Lemma \ref{cor1}. Then, we establish the relationship between the approximate Bellman equation in (\ref{bellman2}) in the  Lemma \ref{cor1} and the PDE in (\ref{bellman3}) in  Theorem \ref{HJB1}.

\emph{1. Relationship between the Equivalent Bellman and the Approximate Bellman Equation:} We establish the following lemma on the approximate Bellman equation to  simplify the equivalent Bellman equation in (\ref{bellman1}):
\begin{lemma}	[Approximate Bellman Equation]	\label{cor1}
For any given weights $\boldsymbol{\beta}$ and $\boldsymbol{\gamma}$, if
\begin{itemize}
\item	there is a unique  ($\theta^\ast, \{ V^\ast\left(\mathbf{Q} \right) \}$) that satisfies the Bellman equation and transversality condition in Theorem~\ref{theorem1}.

\item	there exist $\theta$ and $V\left( \mathbf{Q}\right)$ of class\footnote{$f(\mathbf{x})$ ($\mathbf{x}$ is a $K$-dimensional vector) is of class $\mathcal{C}^2(\mathbb{R}_+^K)$, if the first and second order partial derivatives of $f(\mathbf{x})$ w.r.t. each element of $\mathbf{x}$ are continuous when $\mathbf{x}\in \mathbb{R}_+^K$.} $\mathcal{C}^2(\mathbb{R}_+^K)$ that solve the following \emph{approximate Bellman equation}:
\begin{align}\label{bellman2}
&\theta =\mathbb{E}\bigg[ \min_{ \boldsymbol{\Omega}\left(\boldsymbol{\chi} \right)} \Big[ c\big(\mathbf{Q}, \boldsymbol{\Omega}\big(\boldsymbol{\chi}\big)\big)\hspace{1cm}\forall \mathbf{Q} \in \boldsymbol{\mathcal{Q}}\\
&\hspace{1cm}+\sum_{k=1}^K \frac{\partial V \left(\mathbf{Q} \right) }{\partial Q_k} \Big[  \lambda_k - \sigma_kC_k\big(\mathbf{H},\Omega(\boldsymbol{\chi}) \big)  \Big] \Big]\bigg| \mathbf{Q}\bigg]\nonumber
\end{align}
and for all admissible  control policy $\Omega$, the transversality condition  in (\ref{trans1}) is satisfied for $V$,
\end{itemize}
then, we have
\begin{align}
		\theta^\ast=\theta+o(1), \quad V^\ast\left(\mathbf{Q} \right)=V\left(\mathbf{Q} \right)+o(1), \quad \forall \mathbf{Q} \in \boldsymbol{\mathcal{Q}}
\end{align}
where the error term $o(1)$ asymptotically goes to zero  for sufficiently small slot duration $\tau$.~\hfill\IEEEQED
\end{lemma}

\begin{proof}[Proof of Lemma \ref{cor1}]
Let $\mathbf{Q}' =(Q_1',\cdots, Q_k')= \mathbf{Q}(t+1)$ and $\mathbf{Q}=(Q_1,\cdots, Q_k)=\mathbf{Q}(t)$. For the queue dynamics in (\ref{queueing}) and sufficiently small $\tau$, we have $Q_k'  = Q_k- \sigma_kC_k\left(\mathbf{H},\mathbf{P}\right)  + A_k\tau$, ($\forall k$). Therefore, if  $V\left(\mathbf{Q}\right)$ is of class $\mathcal{C}^2(\mathbb{R}_+^K)$, we have the following Taylor expansion on $V\left( \mathbf{Q}'\right)$:
\begin{align}	
&\mathbb{E}\left[ V\left( \mathbf{Q}'\right) \big| \mathbf{Q} \right]\\
=&V\left( \mathbf{Q}\right)+\sum_{k=1}^K  \frac{\partial V\left(\mathbf{Q}\right)}{\partial Q_k} \Big[  \lambda_k - \mathbb{E}\big[\sigma_kC_k\big(\mathbf{H},\boldsymbol{\Omega}(\boldsymbol{\chi}) \big) \Big|  \mathbf{Q} \Big]   \tau + o(\tau)\nonumber
\end{align}

For notation convenience, let $F_{\boldsymbol{\chi}}(\theta, V, \boldsymbol{\Omega}(\boldsymbol{\chi}))$ denote the \emph{Bellman operator}:
\begin{align}
	  F_{\boldsymbol{\chi}}(\theta, V, \boldsymbol{\Omega}(\boldsymbol{\chi})) =& \sum_{k=1}^K \frac{\partial V \left(\mathbf{Q} \right) }{\partial Q_k} \Big[  \lambda_k - \sigma_kC_k\big(\mathbf{H},\boldsymbol{\Omega}(\boldsymbol{\chi}) \big) \Big] 	\notag \\
	  &-\theta +  c\left(\mathbf{Q}, \boldsymbol{\Omega}\left(\boldsymbol{\chi} \right)\right)+\nu  G_{\boldsymbol{\chi}}(V,\boldsymbol{\Omega}(\boldsymbol{\chi}))
\end{align}
for some smooth function $G_{\boldsymbol{\chi}}$ and $\nu=o(1)$ (w.r.t. $\tau$). Denote
$F_{\boldsymbol{\chi}}(\theta, V)=\min_{ \boldsymbol{\Omega}\left( \mathbf{Q} \right)} F_{\boldsymbol{\chi}}(\theta, V, \boldsymbol{\Omega}(\boldsymbol{\chi}))$.
Suppose $\left(\theta^\ast, V^\ast\right)$ satisfies the Bellman equation in (\ref{bellman1}), we have $\mathbb{E}\left[ F_{\boldsymbol{\chi}}\left( \theta^\ast, V^\ast\right) \big|\mathbf{Q}\right]= \mathbf{0}, \quad \forall \mathbf{Q} \in \boldsymbol{\mathcal{Q}}$.
Similarly, if $\left(\theta, V\right)$ satisfies the approximate Bellman equation in (\ref{bellman2}), we have
\begin{align}	\label{defapp}
	\mathbb{E}\left[F^\dagger_{\boldsymbol{\chi}}\left( \theta, V\right) \big|\mathbf{Q}\right]= \mathbf{0}, \quad \forall \mathbf{Q} \in \boldsymbol{\mathcal{Q}}
\end{align}
where  $F^\dagger_{\boldsymbol{\chi}}(\theta, V)=\min_{ \boldsymbol{\Omega}\left( \mathbf{Q} \right)} F^\dagger_{\boldsymbol{\chi}}(\theta, V, \boldsymbol{\Omega}(\boldsymbol{\chi}))$
and $F^\dagger_{\boldsymbol{\chi}}(\theta, V, \boldsymbol{\Omega}(\boldsymbol{\chi}))= F_{\boldsymbol{\chi}}(\theta, V, \boldsymbol{\Omega}(\boldsymbol{\chi}))-\nu  G_{\boldsymbol{\chi}}(V,\boldsymbol{\Omega}(\boldsymbol{\chi}))$. We then establish  the  following lemma.
\begin{lemma}\label{lemma2}
If $\left(\theta, V\right)$ satisfies the approximate Bellman equation in (\ref{bellman2}), then $\big|\mathbb{E}\big[F_{\boldsymbol{\chi}}(\theta, V)\big|\mathbf{Q}\big]\big| = o(1)$ for any $\mathbf{Q} \in \boldsymbol{\mathcal{Q}}$.~\hfill\IEEEQED
\end{lemma}
\begin{proof}[Proof of Lemma~\ref{lemma2}]
For  any $\boldsymbol{\chi}$, we have $F_{\boldsymbol{\chi}}(\theta, V)=\min_{ \boldsymbol{\Omega}\left( \boldsymbol{\chi} \right)}\big[F^\dagger_{\boldsymbol{\chi}}(\theta, V, \boldsymbol{\Omega}(\boldsymbol{\chi}))+ \nu  G_{\boldsymbol{\chi}}(V,\boldsymbol{\Omega}(\boldsymbol{\chi})) \big] \geq \min_{ \boldsymbol{\Omega}\left( \boldsymbol{\chi} \right)} F^\dagger_{\boldsymbol{\chi}}(\theta, V, \boldsymbol{\Omega}(\boldsymbol{\chi}))+ \nu \min_{ \boldsymbol{\Omega}\left( \boldsymbol{\chi}  \right)} G_{\boldsymbol{\chi}}(V,\boldsymbol{\Omega}(\boldsymbol{\chi})) $. Besieds, $F_{\boldsymbol{\chi}}(\theta, V)\leq \min_{ \boldsymbol{\Omega}\left( \boldsymbol{\chi}  \right)}F^\dagger_{\boldsymbol{\chi}}(\theta, V, \boldsymbol{\Omega}(\boldsymbol{\chi}))+ \nu  G_{\boldsymbol{\chi}}(V,\boldsymbol{\Omega}^\dagger (\boldsymbol{\chi} )) $, where $\boldsymbol{\Omega}^\dagger =\arg \min_{ \boldsymbol{\Omega}\left( \boldsymbol{\chi}  \right)}F^\dagger_{\boldsymbol{\chi}}(\theta, V, \boldsymbol{\Omega}(\boldsymbol{\chi}))$. Since $\mathbb{E}\big[\min_{ \boldsymbol{\Omega}\left( \boldsymbol{\chi}  \right)} F^\dagger_{\boldsymbol{\chi}}(\theta, V, \boldsymbol{\Omega}(\boldsymbol{\chi}))\big|\mathbf{Q}\big]=0$ according to (\ref{defapp}), and $F^\dagger_{\boldsymbol{\chi}}$ and $G_{\boldsymbol{\chi}}$ are all smooth and bounded functions, we have $\big|\mathbb{E}\big[F_{\boldsymbol{\chi}}(\theta, V)\big|\mathbf{Q}\big]\big| = o(1)$ (w.r.t. $\tau$).
\end{proof}

We establish the following lemma to prove Lemma \ref{cor1}.
\begin{lemma}		\label{lemma3}
	Suppose $\mathbf{E}\big[F_{\boldsymbol{\chi}}(\theta^\ast, V^\ast)\big|\mathbf{Q} \big]= 0$ for all $\mathbf{Q}$ together with the transversality condition in (\ref{trans1})  has a unique solution $(\theta^*, V^\ast)$. If $(\theta, V)$ satisfies the approximate Bellman equation in (\ref{bellman2}) and the transversality condition in (\ref{trans1}), then $\theta=\theta^\ast+o\left(1 \right)$, $V\left(\mathbf{Q} \right)=V^\ast\left(\mathbf{Q} \right)+o\left(1 \right)$ for all  $\mathbf{Q} $, where $o(1)$ asymptotically goes to zero as $\tau$ goes to zero.~\hfill\IEEEQED	
\end{lemma}
\begin{proof}	[Proof of Lemma \ref{lemma3}]
Suppose for some $\mathbf{Q}'$, $V\left(\mathbf{Q}' \right)=V^\ast\left(\mathbf{Q}' \right)+\mathcal{O}\left(1 \right)$ (w.r.t. $\tau$). From Lemma~\ref{lemma2}, we have $\big|\mathbb{E}\big[F_{\boldsymbol{\chi}}(\theta, V)\big|\mathbf{Q}\big]\big| = o(1)$ (w.r.t. $\tau$). Letting $\tau \rightarrow 0$, we have $\mathbb{E}\big[F_{\boldsymbol{\chi}}(\theta, V)\big|\mathbf{Q}\big] = 0$ for all $\mathbf{Q}$ and the  transversality condition in (\ref{trans1}). However, $V\left(\mathbf{Q}' \right) \neq V^\ast\left(\mathbf{Q}' \right)$ due to  $V\left(\mathbf{Q}' \right)=V^\ast\left(\mathbf{Q}' \right)+\mathcal{O}\left(1 \right)$. This contradicts the condition that $(\theta^*, V^\ast)$ is a unique solution of $F_{\boldsymbol{\chi}}(\theta^\ast, V^\ast) = 0$ for all $\mathbf{Q}$  and the transversality condition in (\ref{trans1}). Hence, we must have $V\left(\mathbf{Q} \right)=V^\ast\left(\mathbf{Q} \right)+o\left(1 \right)$ for all  $\mathbf{Q} $. Similarly, we can establish $\theta= \theta^\ast + o(1)$.
\end{proof}
\end{proof}

\emph{2. Relationship between the Approximate Bellman Equation and the PDE:}
For notation convenience, we write $J\left(\mathbf{Q}\right)$ in place of $J\left(\mathbf{Q};L^\delta \right)$. It can be  observed that if ($c^{\infty}, \{ J\left(\mathbf{Q} \right) \}$) satisfies  (\ref{bellman3}), it also satisfies  (\ref{bellman2}). Furthermore, since $J\left(\mathbf{Q}\right)=\mathcal{O}(\sum_{k=1}^K Q_k^3)$, then $\lim_{t \rightarrow \infty}\mathbb{E}^{\boldsymbol{\Omega}}\left[J\left(\mathbf{Q}(t)\right) \right]< \infty$ for any admissible policy $\boldsymbol{\Omega}$. Hence, $J\left(\mathbf{Q} \right)=\mathcal{O}(\sum_{k=1}^K Q_k^3 )$ satisfies the transversality condition in (\ref{trans1}). Next, we show that the optimal  policy $\boldsymbol{\Omega}^{J \ast}$ obtained from (\ref{bellman3}) is an admissible control policy according to Definition \ref{admissibledis}.

Define a \emph{Lyapunov function} as $L(\mathbf{Q})=J \left(\mathbf{Q}\right)$. We  define the \emph{conditional queue drift} as $\Delta (\mathbf{Q})= \mathbb{E}^{\boldsymbol{\Omega}^{J \ast}}\big[\sum_{k=1}^K\left(Q_k(t+1)-Q_k(t) \right)\big|\mathbf{Q}(t)=\mathbf{Q} \big]$ and \emph{conditional Lyapunov drift} as $\Delta L(\mathbf{Q})= \mathbb{E}^{\boldsymbol{\Omega}^{J \ast}}\big[L(\mathbf{Q}(t+1))-L(\mathbf{Q}(t))\big|\mathbf{Q}(t) =\mathbf{Q}\big]$.  We first have the following relationship between $\Delta (\mathbf{Q})$ and $\Delta L(\mathbf{Q})$:	
\begin{align}
		\Delta L(\mathbf{Q}) &\geq\mathbb{E}^{\boldsymbol{\Omega}^{J \ast}}\left[  \sum_{k=1}^K \frac{\partial L(\mathbf{Q})}{\partial Q_k}\left(Q_k(t+1)-Q_k(t) \right)	\bigg|\mathbf{Q}(t)=\mathbf{Q}  \right]\nonumber\\ &\overset{(a)}\geq  \Delta (\mathbf{Q})\quad 	 \label{res2222usew}	
	\end{align}
if at least one of $\{Q_k: \forall k\}$ is sufficiently large, where $(a)$ is due to the condition that $\left\{\frac{\partial J\left(\mathbf{Q}; \boldsymbol{\epsilon} \right)}{\partial Q_k}: \forall k \right\}$  are increasing functions of all $Q_k$.

Since $(\lambda_1, \dots, \lambda_K)$ is strictly interior to the stability region $\Lambda$, there exists   $\overline{\boldsymbol{\lambda}}=(\lambda_1+\kappa_1, \dots, \lambda_K+\kappa_K)\in \Lambda$  for some positive $\boldsymbol{\kappa}=\{\kappa_k: \forall k\}$ \cite{comparison2}. From \emph{Corollary 1} of \cite{neelyitp}, there exists a stationary randomized QSI-independent policy $\widetilde{\boldsymbol{\Omega}}$  such that
\begin{align}
	&\hspace{1cm}\sum_{k=1}^K \mathbb{E}^{\widetilde{\boldsymbol{\Omega}}}\left[ \gamma_k P_k \big|\mathbf{Q}(t)=\mathbf{Q}  \right]=\overline{P}(\boldsymbol{\kappa})\nonumber\\
&\mathbb{E}^{\widetilde{\boldsymbol{\Omega}}}\left[ \sigma_k C_k(\mathbf{H}, \mathbf{P})\big|\mathbf{Q}(t)=\mathbf{Q}  \right] \geq \lambda_k+\kappa_k, \quad \forall k		\label{49ersimpor}
\end{align}
where $\overline{P}(\boldsymbol{\kappa})$ is the minimum average power for the system stability  when the arrival rate is $\overline{\boldsymbol{\lambda}}$. The Lyapunov drift $\Delta L(\mathbf{Q})$ is given by
\begin{align}
	& \Delta L(\mathbf{Q})+ \mathbb{E}^{\boldsymbol{\Omega}^{J \ast}} \left[ \sum_{k=1}^K  \gamma_k P_k \tau  \bigg|\mathbf{Q}(t)=\mathbf{Q}  \right] \notag \\
	\approx  & \sum_{k=1}^K \frac{\partial L(\mathbf{Q})}{\partial Q_k} \lambda_k \tau  \nonumber\\
&+ \mathbb{E}^{\boldsymbol{\Omega}^{J \ast}}\left[  \sum_{k=1}^K  \left(\gamma_k P_k  \tau -\frac{\partial L(\mathbf{Q})}{\partial Q_k}\sigma_kC_k(\mathbf{H}, \mathbf{P})\tau\right)\bigg|\mathbf{Q}(t)=\mathbf{Q}  \right] \notag \\
	\overset{(b)}\leq  & \sum_{k=1}^K \frac{\partial L(\mathbf{Q})}{\partial Q_k} \lambda_k \tau  \nonumber\\
&+ \mathbb{E}^{\widetilde{\boldsymbol{\Omega}}}\left[  \sum_{k=1}^K  \left(\gamma_k  P_k  \tau -\frac{\partial L(\mathbf{Q})}{\partial Q_k}\sigma_kC_k(\mathbf{H}, \mathbf{P})\tau\right)\bigg|\mathbf{Q}(t)=\mathbf{Q}  \right]\notag \\
	\overset{(c)} \leq  & -\sum_{k=1}^K \frac{\partial L(\mathbf{Q})}{\partial Q_k} \kappa_k \tau  +\overline{P}(\boldsymbol{\kappa}) \tau  \label{asdadaappne}
\end{align}
if at least one of $\{Q_k: \forall k\}$ is sufficiently large, where $(b)$ is due to $\boldsymbol{\Omega}^{J \ast}$ achieves the minimum of  (\ref{bellman3}) and $(c)$ is due to (\ref{49ersimpor}).  Combining (\ref{asdadaappne}) with (\ref{res2222usew}), we have $\Delta (\mathbf{Q})\leq \Delta L(\mathbf{Q}) \leq  -\sum_{k=1}^K \frac{\partial L(\mathbf{Q})}{\partial Q_k} \kappa \tau  +\overline{P}(\boldsymbol{\kappa}) \tau	 <0 $ if at least one of $\{Q_k: \forall k\}$ is sufficiently large. Therefore, $\mathbb{E}\big[A_k-G_k(\mathbf{H},\boldsymbol{\Omega}^{J \ast}(\boldsymbol{\chi}))   \big|\mathbf{Q}\big]< 0$ when $Q_k > \overline{Q}_k$ for some large $\overline{Q}_k$. Let $\phi_k(r,\mathbf{Q})=\ln \big(\mathbb{E}\big[e^{\left(A_k-G_k(\mathbf{H},\boldsymbol{\Omega}^{J \ast}(\boldsymbol{\chi})  ) \right)r}\big| \mathbf{Q} \big] \big)$ be the \emph{semi-invariant moment generating function} of $A_k-G_k\big(\mathbf{H},\boldsymbol{\Omega}^{J \ast}(\boldsymbol{\chi})\big)$.  Then,  $\phi_k(r,\mathbf{Q})$ will have a unique positive root $r_k^\ast(\mathbf{Q})$ ($\phi_k(r_k^\ast(\mathbf{Q}),\mathbf{Q})=0$) \cite{dspgal}.  Let $r_k^\ast= r_k^\ast(\overline{\mathbf{Q}})$, where $\overline{\mathbf{Q}}=(\overline{Q}_1, \dots, \overline{Q}_K)$. Using the Kingman bound \cite{dspgal} result that $F_k(x) \triangleq \Pr\big[ Q_k \geq x \big] \leq e^{-r_k^\ast x} $, if $x \geq \overline{x}_k$ for sufficiently large $\overline{x}_k$, we have
\begin{align}
	 &\mathbb{E}^{\boldsymbol{\Omega}^{J \ast}} \left[J\left(\mathbf{Q}\right) \right] \nonumber\\
\leq& C \sum_{k=1}^K \mathbb{E}^{\boldsymbol{\Omega}^{J \ast}} \left[ Q_k^3  \right]=C\sum_{k=1}^K \left[\int_0^{\infty} \Pr \left[Q_k^3 >s \right] \mathrm{d}s \right] 	\notag \\
	 \leq & C \sum_{k=1}^K \left[ \int_{0}^{\overline{x}_k^3}F_k(s^{1/3}) \mathrm{d}s  + \int_{\overline{x}_k^3}^{\infty} F_k(s^{1/3})\mathrm{d}s \right] \nonumber\\
\leq&  C \sum_{k=1}^K \left[\overline{x}_k^3+ \int_{\overline{x}_k^3}^{\infty}  e^{-r_k^\ast s^{1/3}}  \mathrm{d}s \right]	 < \infty
\end{align}
for some constant $C$. Therefore, $\boldsymbol{\Omega}^{J \ast}$ is an admissible control policy and we have  $V \left(\mathbf{Q} \right)=J\left(\mathbf{Q}\right)$ and $\theta=c^\infty$.

Combining Corollary \ref{cor1}, we have $V^\ast\left(\mathbf{Q} \right)=J\left(\mathbf{Q}\right)+o(1)$ and $\theta^\ast=c^\infty+o(1)$ for sufficiently small $\tau$.

\section*{Appendix C: Proof of Lemma~\ref{theorem3}}	
We first prove that $ J \left(\mathbf{Q};0\right)=\sum_{k=1}^K J_k\left(Q_k \right)$. The PDE in (\ref{bellman3}) for the base system  is
\begin{align}\label{sumperHJBapp}	
&\mathbb{E}\bigg[ \min_{\boldsymbol{\Omega}(\boldsymbol{\chi})}  \bigg[\sum_{k=1}^K \bigg( \beta_k \frac{Q_k}{\lambda_k}+\gamma_k P_k \\
&\hspace{1cm}+ \frac{\partial  J \left(\mathbf{Q};0\right)} {\partial Q_k} \Big(\lambda_k - \sigma_k C_k\big(\mathbf{H},\mathbf{P} \big)\Big) \bigg)\bigg]\bigg| \mathbf{Q} \bigg] -c^\infty=0\nonumber
\end{align}
We have the following lemma to prove the decomposable structures of $J \left(\mathbf{Q};0\right)$ and $c^\infty$ in (\ref{sumperHJBapp}).
\begin{lemma}	[Decomposed Optimality Equation]	\label{decomplem}
	Suppose there exist $c_k^\infty$ and $J_k \left(Q_k \right) \in \mathbb{C}^2\left(\mathbb{R}_+ \right)$ that solve the following per-flow optimality equation (PFOE):
	\begin{align}\label{perflowhjb}
&\mathbb{E}\bigg[  \min_{P_k \geq 0}\bigg[\beta_k \frac{Q_k}{\lambda_k}+\gamma_k P_k\\
 &\hspace{1cm}+ J_k'(Q_k) \Big(\lambda_k - \sigma_k C_k^0\big({H}_{kk},P_k \big)\Big) \bigg]\bigg|Q_k \bigg] -c_k^\infty=0\nonumber
	\end{align}
	where $C_k^0\big({H}_{kk},P_k \big)=\log_2 \big(1+\frac{1}{\Gamma}\frac{H_{kk}P_k}{N_0} \big)$. Then, $ J\left( \mathbf{Q};0\right)=\sum_{k=1}^K J_k\left(Q_k \right)$ and $c^\infty=\sum_{k=1}^K c_k^\infty$ satisfy  (\ref{sumperHJBapp}).~\hfill\IEEEQED
\end{lemma}

Lemma \ref{decomplem} can be proved using the fact that the dynamics of the $K$ queues at the transmitters are decoupled when $L^\delta=0$. The details are omitted for conciseness.

Next, we solve the PFOE in  (\ref{perflowhjb}). The optimal transmit power from (\ref{perflowhjb}) is given by
\begin{equation}    \label{acoeopt}
P_k^*=\left(\frac{J_k'\left(Q_k \right) \sigma_k}{\gamma_k \ln2}-\frac{\Gamma N_0}{H_{kk}}\right)^+
\end{equation}
Substituting the optimal transmit power $P_k^\ast$ to (\ref{perflowhjb}), and using the fact that $\sigma$ follows a Bernoulli distribution with mean $\frac{1}{|\mathcal{N}_k|+1}$ (from Assumption \ref{a_MAC}) and  $H_{kk}$ follows a negative exponential distribution with mean $L_{kk}$ (from Assumption \ref{a_CSI}), we calculate the expectations in (\ref{perflowhjb}) as follows:
\begin{align}\label{calD1}
	&\mathbb{E}\left[ \gamma_k P_k^\ast \big| Q_k \right] \nonumber\\
=&\frac{1}{(|\mathcal{N}_k|+1) L_{kk}}\int_{\frac{N_0\Gamma \gamma_k \ln2}{J_k'\left(Q_k \right)}}^\infty \left(\frac{J_k'\left(Q_k \right)}{ \ln2}-\frac{N_0\Gamma \gamma_k }{x}\right) e^{-x/L_{kk}}\mathrm{d}x \notag \\
    =& \frac{1}{|\mathcal{N}_k|+1}\Bigg(\frac{J_k'\left(Q_k \right)}{\ln2}e^{-\frac{N_0\Gamma \gamma_k \ln2}{J_k'\left(Q_k \right)L_{kk}}}\nonumber\\
&\hspace{2cm}-\frac{\gamma_k  N_0\Gamma}{L_{kk}}E_1\left(\frac{N_0\Gamma \gamma_k \ln2}{J_k'\left(Q_k \right)L_{kk}}\right)\Bigg)
\end{align}
Using the same integration region, we have
\begin{align}\label{calD2}
	& \mathbb{E}\left[\sigma_k\log_2\big(1+P_k^*H_{kk}/(\Gamma N_0)\big) \big| Q_k \right]\nonumber\\
=& \frac{1}{(|\mathcal{N}_k|+1) \ln 2}E_1\left(\frac{N_0\Gamma \gamma_k \ln 2}{J_k'\left(Q_k \right)L_{kk}}\right)
\end{align}
where  $E_1(z)  \triangleq \int_z^{\infty} \frac{e^{-t}}{t}\mathrm{d}t $ is the exponential integral function. We then calculate ${c_k^{\infty}}$. Since (\ref{perflowhjb}) should hold when $Q_k=0$, we have $\hspace{1cm}c_k^\infty=\mathbb{E}\left[ \gamma_k P_k^\ast \big| Q_k=0 \right]$ and $\mathbb{E}\left[\sigma_k\log_2\big(1+P_k^*H_{kk}/(\Gamma N_0)\big) \big| Q_k=0 \right]=\lambda_k$. Substituting these into (\ref{calD1}), we can calculate $c_k^\infty$ as shown in Lemma~\ref{theorem3}. Substituting (\ref{calD1}), (\ref{calD2}), and $c_k^\infty$ into (\ref{perflowhjb}) and letting $a_k \triangleq \frac{N_0\gamma_k\ln 2}{L_{kk}}$, we have the following ODE:
\begin{align}\label{finalode}	
&    \beta_k \frac{Q_k}{\lambda_k}+\frac{1}{|\mathcal{N}_k|+1}\Bigg(\frac{J_k'\left(Q_k \right)}{\ln2}e^{-\frac{N_0\Gamma \gamma_k \ln2}{J_k'\left(Q_k \right)L_{kk}}} \nonumber\\
&-\frac{\gamma_k  N_0 \Gamma}{L_{kk}}E_1\left(\frac{N_0\Gamma \gamma_k \ln2}{J_k'\left(Q_k \right)L_{kk}}\right)\Bigg)-{c_k^{\infty}}+ J_k'\left(Q_k \right)\lambda_k   \notag \\
&    - J_k'\left(Q_k \right)\frac{1}{(|\mathcal{N}_k|+1) \ln 2}E_1\left(\frac{N_0\Gamma \gamma_k \ln 2}{J_k'\left(Q_k \right)L_{kk}}\right) =0
\end{align}
According to Section 0.1.7.3 of \cite{HandODE}, we can obtain the parametric solution of (\ref{finalode}) as shown in (\ref{perflow}) in Lemma~\ref{theorem3}.

\section*{Appendix D: Proof of Corollary~\ref{cor2}}	
First, we obtain the highest order term of $J_k\left(Q_k \right)$. The series expansions of $E_1(x)$ and $e^x$ are given by
\begin{equation}	\label{asypototoc}
E_1(x) = -\gamma_{eu} - \ln x - \sum_{n=1}^{\infty} \frac{\left(-x \right)^n}{n ! n}, \quad e^x =\sum_{n=0}^{\infty} \frac{x^n}{n!}
\end{equation}
Using (\ref{asypototoc}), (\ref{perflow})  induces that  $Q_k(y)=\mathcal{O}(y\ln y)$ and $J_k(y)=\mathcal{O}(y^2\ln y)$ as $y\rightarrow \infty$. In other words, we have $\delta_1   y\ln y \leq Q_k(y) \leq \delta_1' y\ln y$ when $y \rightarrow \infty$ for some constants $\delta_1$ and $\delta_1'$, and $\delta_2 y^2\ln y \leq J_k(y) \leq \delta_2' y^2\ln y$ when $y \rightarrow \infty$ for some constants $\delta_2$ and $\delta_2'$ .Therefore,
\begin{align}
\delta_2\left(\frac{Q_k/\delta'_1}{W(Q_k/\delta'_1)}\right)^2 \ln\left(\frac{Q_k/\delta'_1}{W(Q_k/\delta'_1)}\right) \leq J_k(y) \nonumber\\
\leq \delta'_2\left(\frac{Q_k/\delta_1}{W(Q_k/\delta_1)}\right)^2 \ln\left(\frac{Q_k/\delta_1}{W(Q_k/\delta_1)}\right)
\end{align}
where $W$ is the \emph{Lambert} function \cite{lambert}. Since $W(x)=\mathcal{O}(\ln x)$ for sufficiently large $x$ \cite{lambert}, we conclude that $J_k \left( Q_k \right) = \mathcal{O} \left(\frac{Q_k^2}{\ln{Q_k}} \right)$   as $Q_k \rightarrow \infty$.

Next, we obtain the coefficient of the highest order term $\frac{Q_k^2}{\ln{Q_k}}$. Using (\ref{asypototoc}), the PFOE equation in (\ref{finalode}) implies
\begin{align}\label{bellmanE1}
	 J_k'\left(Q_k \right)\ln\big(J_k'\left(Q_k \right)\big) = \frac{\beta_k (|\mathcal{N}_k|+1) \ln 2}{\lambda_k } Q_k + o(Q_k)
\end{align}
Since $J_k\left(Q_k \right) = \mathcal{O} \left(\frac{Q_k^2}{\ln \left( Q_k \right)}  \right)$, there exist constants $\delta$ and $\delta'$ such that
\begin{align}
	& \delta \frac{Q_k^2}{\ln \left( Q_k \right)}  \leq  J_k\left(Q_k \right) \leq \delta' \frac{Q_k^2}{\ln \left( Q_k \right)} 	\notag \\
	\Rightarrow  &   \Delta \frac{Q_k}{\ln \left( Q_k \right)}	\leq J_k'\left(Q_k \right) \leq  \Delta' \frac{Q_k}{\ln \left( Q_k \right)} 	\notag \\
	\Rightarrow & \ln \left(   \Delta \right) +  \ln \left( Q_k \right) - \ln \ln \left( Q_k \right)	\leq \ln \left( J_k'\left(Q_k \right) \right) \notag \\
&\leq \ln \left( \Delta' \right) + \ln \left( Q_k \right) - \ln \ln \left( Q_k \right)\notag \\
	\Rightarrow&   \Delta  Q_k + o\left( Q_k\right) \leq J_k'\left(Q_k \right) \ln \left( J_k'\left(Q_k \right) \right) \leq  \Delta'  Q_k + o\left( Q_k\right)
\end{align}
 where  $\Delta$ and $\Delta'$  are some constants that are independent of the system parameters.  Comparing it with (\ref{bellmanE1}), we have  $\Delta, \Delta' \propto  \frac{\beta_k (|\mathcal{N}_k|+1) \ln 2}{\lambda_k }	 \Rightarrow    \delta,  \delta' \propto  \frac{\beta_k (|\mathcal{N}_k|+1) \ln 2}{2\lambda_k }$, where $x \propto y$ means that $x$ is proportional to $y$. Finally, we conclude that $J_k\left(Q_k \right) = \frac{\beta_k (|\mathcal{N}_k|+1)}{2\lambda_k }\frac{Q_k^2}{\log_2 \left( Q_k \right)} +o\left(\frac{Q_k^2}{\log_2 \left( Q_k \right)} \right)$ and $J_k'\left(Q_k \right) = \frac{\beta_k (|\mathcal{N}_k|+1) }{\lambda_k }\frac{Q_k}{\log_2 \left( Q_k \right)}+o\left(\frac{Q_k}{\log_2 \left( Q_k \right)} \right)$.

\section*{Appendix E: Proof of Theorem \ref{theorem4}}	
We first write $H_{kj}=L_{kj}\widetilde{H}_{kj}$, where $\widetilde{H}_{kj}$ is the short-term fading path gain. Taking the first order Taylor expansion of the L.H.S. of the PFOE in (\ref{bellman3}) at $L_{kj}=0$ ($\forall k \neq j, d_{kj}>\delta$), $P_k=P_k^{\ast}$ (where $P_k^\ast$ minimize the L.H.S. of (\ref{perflowhjb})), and using  parametric optimization analysis \cite{paraanay},  we have the following result regarding the approximation error:
\begin{align}
J\left(\mathbf{Q}; L^\delta \right)-J\left(\mathbf{Q};  0  \right) =\sum_{i=1}^K \sum_{j \neq i, \atop j \notin \mathcal{N}_i(\delta)} L_{ij}  \widetilde{J}_{ij}(\mathbf{Q})+ \mathcal{O}((L^\delta)^2)	\label{tayloee}
\end{align}
where we have  $L^\delta=\mathcal{O}(\frac{1}{\delta^2})$ according to  Assumption \ref{long}. $\widetilde{J}_{ij}(\mathbf{Q})$  captures the coupling terms in $J\left(\mathbf{Q}\right)$  satisfying:
\begin{small}\begin{align}	
	\hspace{-0.4cm} \sum_{k=1}^K \left(\lambda_k  - \mathbb{E} \left[\left.  \sigma_k \log_2\left(1+\frac{P_k L_{kk} \widetilde{H}_{kk}}{\Gamma N_0}\right) \right| \mathbf{Q}  \right]  \right) \frac{\partial \widetilde{J}_{ij}\left(\mathbf{Q}\right) }{\partial Q_k} \nonumber\\
  + \mathbb{E} \left[\left.\frac{J_i'\left(Q_i \right)}{\ln 2}   \frac{\sigma_i P_i^{\ast}L_{ii}\widetilde{H}_{ii}\widetilde{H}_{ij}}{\Gamma N_0^2+N_0P_i^{\ast} L_{ii}\widetilde{H}_{ii}}  P_j^{\ast} \right| \mathbf{Q}  \right] =\widetilde{\theta}_{ij}  \label{PDEappd}
\end{align}\end{small}with boundary condition $\widetilde{J}_{ij}\left(\mathbf{Q} \right)\big|_{Q_i = 0}=0$ or $\widetilde{J}_{ij}\left(\mathbf{Q}  \right)\big|_{Q_j =0}=0$, and $\widetilde{\theta}_{kj}=\frac{\partial \theta \left(L \right)}{\partial L_{kj}}$ is constant (where we treat $\theta$  as a function of $\left\{L_{ij}: \forall i \neq j \right\}$). According to (\ref{calD1}) and (\ref{calD2}), we have
\begin{small}\begin{align}
\mathbb{E} \left[\left. \sigma_k \log_2\left(1+\frac{P_k L_{kk} \widetilde{H}_{kk}}{N_0}\right) \right| \mathbf{Q}  \right]= \frac{1}{(|\mathcal{N}_k(\delta)|+1) \ln 2}\mathcal{O}\left(\ln Q_k \right)\nonumber
\end{align}
\begin{align}
&\mathbb{E} \left[\frac{J_i'\left(Q_i \right)}{\ln 2} \left.\frac{\sigma_i P_i^{\ast}L_{ii}\widetilde{H}_{ii}\widetilde{H}_{ij}}{\Gamma N_0^2+N_0P_i^{\ast} L_{ii}\widetilde{H}_{ii}}  P_j^{\ast} \right| \mathbf{Q}  \right]\nonumber\\
=&\frac{\mathcal{O}\left(J_i'(Q_i)J_j'(Q_j)\right)}{(|\mathcal{N}_i(\delta)|+1) (|\mathcal{N}_j(\delta)|+1)  (\ln 2)^2 \gamma_j N_0}\nonumber\\
=&\frac{\beta_i \beta_j}{(\ln 2)^2 \lambda_i \lambda_j \gamma_j N_0}\mathcal{O}\left(\frac{Q_iQ_j}{\log_2 Q_i \log_2 Q_j}\right)\nonumber
\end{align}\end{small}
Substituting these calculation results into (\ref{PDEappd}), using \emph{3.8.4.7}  of \cite{handbookPDE} and taking into account the boundary conditions, we obtain that $\widetilde{J}_{ij}\left(\mathbf{Q}\right) = D_{ij}\mathcal{O}\left(\frac{Q_i^2 Q_j}{(\log_2\left(Q_i \right))^2\log_2\left(Q_j \right)}\right)$, where $D_{ij}=\frac{\beta_i\beta_j(|\mathcal{N}_i(\delta)|+1)}{2 (\ln2)    \lambda_i \lambda_j \gamma_j N_0}$. Substituting it to (\ref{tayloee}), we obtain the approximation error in Theorem~\ref{theorem4}.

\section*{Appendix F: Proof of Corollary \ref{collaryAlg}}
According to the definition of $\sigma_k$,  the problem in (\ref{utility}) is equivalent to
\begin{align}	\label{obj11}
	\min _{\mathbf{P} }  \ \sum_{k \in \mathcal{A}(\delta)}  \Big(\gamma_k  \sigma_k P_k  - \frac{\partial \widetilde{V}\left(\mathbf{Q}\right)}{\partial Q_k}   C_k\left(\mathbf{H}, \left\{\sigma_k P_k:  k \in \mathcal{A}(\delta)\right\}\right)  \Big)
\end{align}
where $\mathcal{A}(\delta)=\{k: \sigma_k=1\}$ is the set of active transmitters for a  given $\delta$, $ C_k\left(\mathbf{H}, \left\{\sigma_k P_k:  k\in \mathcal{A}(\delta)\right\}\right) =\log_2\left(1+\frac{1}{\Gamma}\frac{H_{kk}\sigma_k P_k}{N_0+\sum_{j\neq k, j \in \mathcal{A}(\delta)}H_{kj} \sigma_j P_j}\right)$.  Denote the objective function in (\ref{obj11}) as $ f\left( \mathbf{P}, L^\delta \right)$. We have the following lemma on the convexity for $f\left( \mathbf{P}, L^\delta \right)$.
\begin{lemma}	[Convexity of $f\left( \mathbf{P}, L^\delta \right)$ for Sufficiently Small $L^\delta$]	\label{lemma8}
	$f\left( \mathbf{P}, L^\delta \right)$ is a convex function of $\mathbf{P}=\left\{P_k:  k\in \mathcal{A}(\delta) \right\}$ when $L^\delta$ is sufficiently small.~\hfill\IEEEQED
\end{lemma}
\begin{proof}
	We adopt the following argument to prove the convexity  \cite{boyd}: given  two  feasible points $\mathbf{x}_1$ and $\mathbf{x}_2$, define $g(t)= f(t \mathbf{x}_1+(1-t)\mathbf{x}_2)$, $0\leq t \leq 1$, then $f(\mathbf{x})$ is a convex function of $\mathbf{x}$ if and only if $g(t)$ is a convex function of $t$, which is equivalent to $\frac{\mathrm{d}^2 g(t)}{\mathrm{d}t^2} \geq 0$ for $0 \leq t \leq 1$.
	
	Consider the convex combination of two  feasible solutions $\mathbf{P}^{(1)}=\big\{P_k^{(1)}:  k\in \mathcal{A}(\delta) \big\}$ and $\mathbf{P}^{(2)}=\big\{P_k^{(2)}:  k\in \mathcal{A}(\delta) \big\}$ as follows: $\mathbf{P}^c=\big\{P_k^c=t P_k^{(1)} + (1-t) P_k^{(2)} :  k\in \mathcal{A}(\delta) \big\}$ and $0 \leq t \leq 1$. We write $H_{kj}=L_{kj}\widetilde{H}_{kj}$, where $\widetilde{H}_{kj}$ is the short-term fading path gain. Denote $\mathbf{P}_{-k}=\{P_j: \forall j \neq k, j \in \mathcal{A}(\delta)\}$, $R_k(\mathbf{P}_{-k}) = N_0+\sum_{j\neq k, j \in \mathcal{A}(\delta)}L_{kj}\widetilde{H}_{kj}P_j$ and $a_k=\frac{1}{\ln 2}\frac{\partial \widetilde{V}\left(\mathbf{Q}\right)}{\partial Q_k} \geq 0$, then the second order derivative of $f\left( \mathbf{P}^c, L^\delta \right)$ is calculated as:
\begin{small}\begin{align}
		& \frac{\mathrm{d}^2 f\left( \mathbf{P}^c, L^\delta \right)}{\mathrm{d}t^2} =\sum_{k \in \mathcal{A}(\delta)} \left[a_k\left( \left( R_k(\mathbf{P}_{-k}^c)  + \frac{1}{\Gamma}L_{kk} \widetilde{H}_{kk}\sigma_k P_k^c \right)^{-2}\right.\right.\nonumber\\
&\hspace{1cm}\left( \frac{\mathrm{d}R_k(\mathbf{P}_{-k}^c)}{\mathrm{d}t}  +  \frac{1}{\Gamma}L_{kk}  L_{kk} \widetilde{H}_{kk}\sigma_k \left(P_k^{(1)}-P_k^{(2)} \right)\right)^2 \notag \\
&\hspace{1cm} \left.\left. - R_k^{-2}(\mathbf{P}_{-k}^c) \left(\frac{\mathrm{d}R_k(\mathbf{P}_{-k}^c)}{\mathrm{d}t}\right)^2 \right)\right]	\label{firstrermsd}
\end{align}\end{small}
where $\frac{\mathrm{d}R_k(\mathbf{P}_{-k}^c) }{\mathrm{d}t}=\sum_{j\neq k, j \in \mathcal{A}(\delta)}L_{kj}\widetilde{H}_{kj}\left(P_j^1-P_j^2 \right)$ does not depend on $t$.	
	
	As $L^\delta$ becomes sufficiently small,  $\frac{\mathrm{d}R_k(\mathbf{P}_{-k}^c)}{\mathrm{d}t} $ is proportional to $L^\delta$ and $\frac{\mathrm{d}R_k(\mathbf{P}_{-k}^c)}{\mathrm{d}t}  + \frac{1}{\Gamma}L_{kk} \widetilde{H}_{kk}\sigma_k \big(P_k^{(1)}-P_k^{(2)} \big)$ is dominate by $ \frac{1}{\Gamma}L_{kk} \widetilde{H}_{kk}\sigma_k \big(P_k^{(1)}-P_k^{(2)} \big)$.  $R_k^{-2}(\mathbf{P}_{-k}^c) \big(\frac{\mathrm{d}R_k(\mathbf{P}_{-k}^c)}{\mathrm{d}t}\big)^2 $ is proportional to $(L^\delta)^2$, and hence it has little impact and can be ignored. Therefore, we have
\begin{small}\begin{align}
&\frac{\mathrm{d}^2 f\left( \mathbf{P}^c, L^\delta \right)}{\mathrm{d}t^2} \approx \sum_{k \in \mathcal{A}(\delta)} \Bigg(a_k \left( R_k(\mathbf{P}_{-k}^c) +  \frac{1}{\Gamma} L_{kk} \widetilde{H}_{kk}\sigma_k P_k^c \right)^{-2} \notag\\
&~~~~\left( \frac{\mathrm{d}R_k(\mathbf{P}_{-k}^c)}{\mathrm{d}t}  + \frac{1}{\Gamma} L_{kk} \widetilde{H}_{kk}\sigma_k \left(P_k^{(1)}-P_k^{(2)} \right)\right)^2 \Bigg)	\geq 0
\end{align}\end{small}
for sufficiently small $L^\delta$. Therefore, $f\left( \mathbf{P}, L^\delta \right)$ is convex for sufficiently small $L^\delta$.
\end{proof}

For sufficiently large $\delta$, $L^\delta$ is sufficiently small, so the problem in (\ref{obj11}) is  convex, and hence (\ref{utility}) is  convex according to Lemma \ref{lemma8}. Furthermore, since the limiting point $\mathbf{P}(\infty)$ of algorithm \ref{algorithm} is  a stationary point of the problem (\ref{utility}), it is also the unique global optimal point of (\ref{utility}).

\begin{IEEEbiography}[{\includegraphics[width=1in,height=1.25in,clip,keepaspectratio]{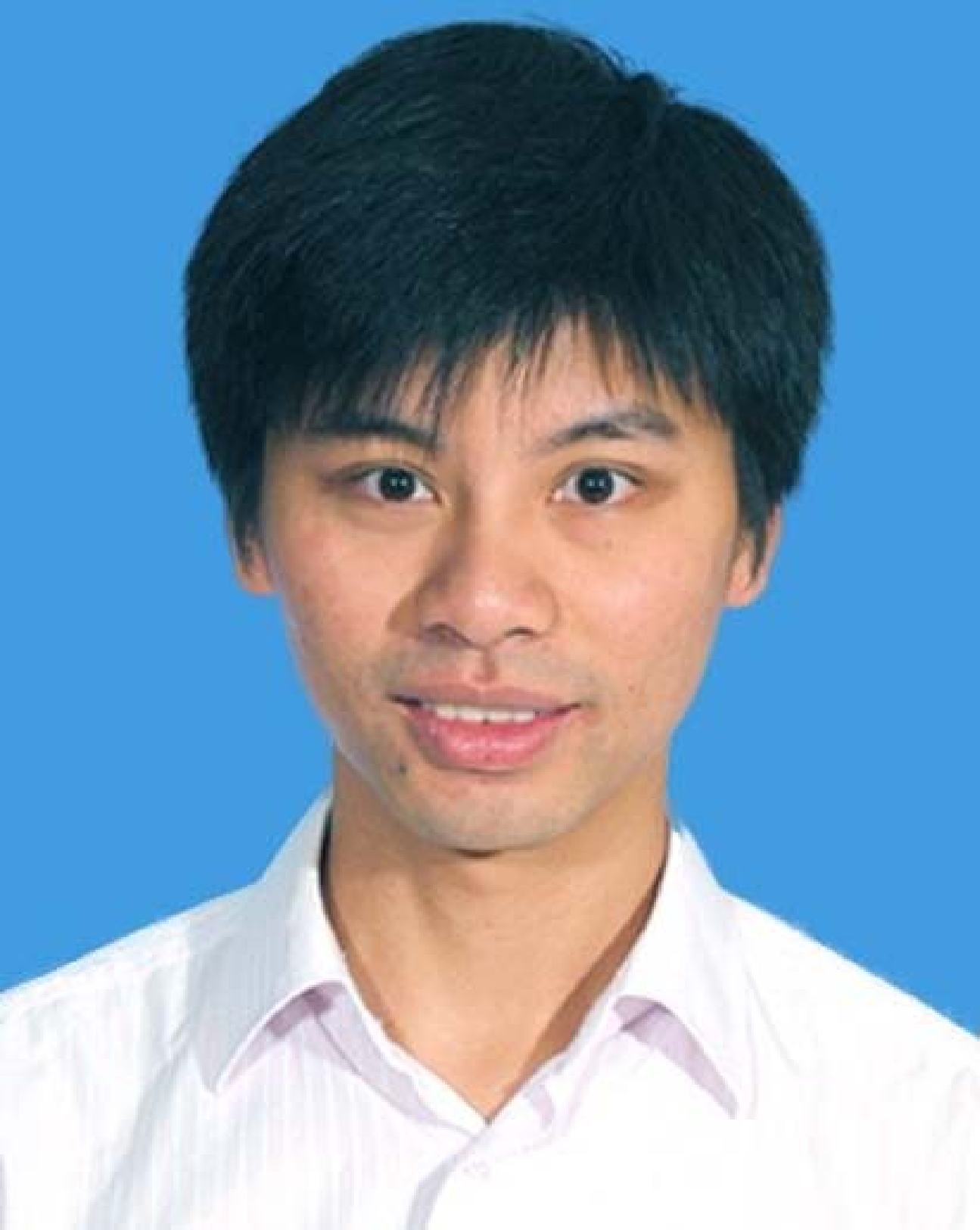}}]{Wei~Wang~(StM'08-M'10)}
received the B.S. degree in Communication Engineering and the Ph.D. degree in Signal and Information Processing from Beijing University of Posts and Telecommunications, China in 2004 and 2009, respectively. Now, he is an associate professor with Department of Information Science and Electronic Engineering, Zhejiang University, China. From Sept. 2007 to Sept. 2008, he was a visiting student with University of Michigan, Ann Arbor, USA. Since Feb. 2013, he has also been a Hong Kong Scholar with Hong Kong University of Science and Technology, Hong Kong. His research interests mainly focus on cognitive radio networks, green communications, and radio resource allocation for wireless networks.

He is the editor of the book ``\emph{Cognitive Radio Systems}" (Intech, 2009) and serves as an editor for \emph{Transactions on Emerging Telecommunications Technologies (ETT)}. He serves as TPC co-chair for CRNet 2010 and NRN 2011, symposium co-chair for WCSP 2013, tutorial co-chair for ISCIT 2011, and also serves as TPC member for major international conferences.
\end{IEEEbiography}

\begin{IEEEbiography}[{\includegraphics[width=1in,height=1.25in,clip,keepaspectratio]{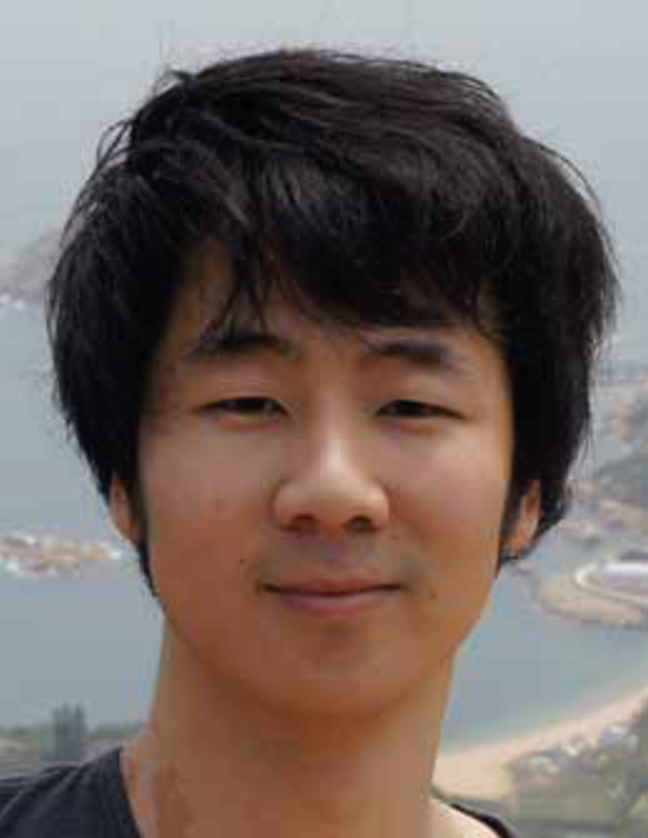}}]{Fan~Zhang~(StM'10)}
revived the B.Eng. (First Class Hons) from Chu Kochen Honors College at Zhejiang University in  2010. He is currently pursuing a Ph.D degree in the Department of Electronic and Computer Engineering, Hong Kong University of Science and Technology (HKUST). His research interests include cross-layer delay-sensitive resource allocation, dynamic programming and control for wireless communication systems.
\end{IEEEbiography}

\begin{IEEEbiography}[{\includegraphics[width=1in,height=1.25in,clip,keepaspectratio]{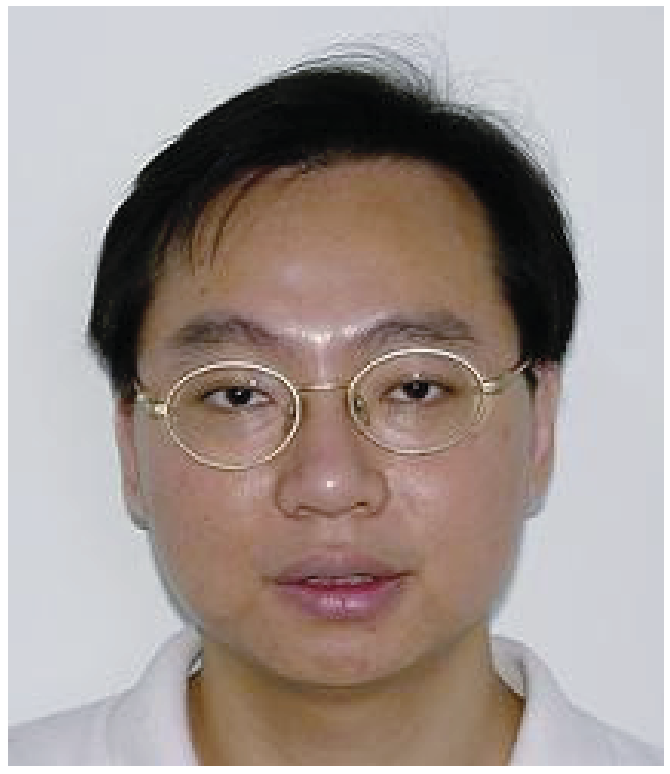}}]{Vincent~K.~N.~Lau~(F'11)}
 received the B.Eng (Distinction 1st Hons) from the University of Hong Kong (1989-1992)
and Ph.D. from Cambridge University (1995-1997). He was with HK Telecom (PCCW) as system
engineer from 1992-1995 and Bell Labs - Lucent Technologies as member of technical staff
from 1997-2003. He is currently a Chair Processor in the Department of ECE, Hong Kong University of Science and Technology (HKUST), and the Founding Director of Huawei-HKUST Joint Innovation Lab. His research interests include the delay-sensitive cross-layer optimization of MIMO/OFDM wireless systems, cooperative communications, as well as stochastic approximation and Markov Decision Process.

He is a Fellow of IEEE, Fellow of HKIE, Changjiang Chair Professor and the Croucher Senior Research Fellow. He is currently an Area Editor of \emph{IEEE Transactions on Wireless Communications}, an Area Editor of \emph{IEEE Signal Processing Letters}.
\end{IEEEbiography}

\end{document}